\UseRawInputEncoding
\pdfoutput=1
\documentclass[10pt,journal]{IEEEtran}
\usepackage{amssymb}
\usepackage{color}
\usepackage{amsmath,amsfonts}
\usepackage{amsthm}
\usepackage{algorithmic}
\usepackage{algorithm}
\usepackage{array}
\usepackage[caption=false,font=normalsize,labelfont=sf,textfont=sf]{subfig}
\usepackage{textcomp}
\usepackage{stfloats}
\usepackage{url}
\usepackage{verbatim}
\usepackage{graphicx}
\usepackage{cite}
\usepackage{epstopdf}
\usepackage{hyperref}
\usepackage{enumerate}
\usepackage{amsmath}
\newtheorem{mylemma}{Lemma}
\newtheorem{mytheorem}{Theorem}
\newtheorem{mycorollary}{Corollary}

\newenvironment*{myproof}{\proof}{\endproof}
\makeatletter
\renewcommand{\maketag@@@}[1]{\hbox{\m@th\normalsize\normalfont#1}}%
\makeatother

\begin{document}

\title{Mean Age of Information in Partial Offloading Mobile Edge Computing Networks}

\author{Ying Dong, Hang Xiao, Haonan Hu,~\IEEEmembership{Member,~IEEE}, Jiliang Zhang, ~\IEEEmembership{Senior Member,~IEEE}, Qianbin Chen,~\IEEEmembership{Senior Member,~IEEE}, Jie Zhang,~\IEEEmembership{Senior Member,~IEEE}
       
\thanks{This work was supported in part by the National Natural Science Foundation of China (NSFC) under Grant 61901075, in part by the UKRI RISIR project under Grant no. 101109294, in part by the Science and Technology Research Program of Chongqing Municipal Education Commission under Grant no. KJZDK202200604, in part by the MSCA IPOSEE project under Grant no. 101086219. \textit{(Ying Dong and Hang Xiao contributed equally to this
 work.) (Corresponding author: Haonan Hu.)}

Ying Dong, Hang Xiao, Haonan Hu, and Qianbin Chen are with the School of Communication and Information Engineering, Chongqing University of Posts and Telecommunications, Chongqing 400065, China. (email: yingd.cqupt@qq.com; cqupt\_xiaohang@foxmail.com; huhn@cqupt.edu.cn; chenqb@cqupt.edu.cn)

Jiliang Zhang is with the College of Information Science and Engineering, Northeastern University, Shenyang 110819, China (e-mail: zhangjiliang1@mail.neu.edu.cn).
		
Jie Zhang is with the Department of Electronic and Electrical Engineering, the University of Sheffield, S1 3JD Sheffield, U.K. (email: jie.zhang@sheffield.ac.uk)
}} 
\maketitle

\begin{abstract} 
The age of information (AoI) performance analysis is essential for evaluating the information freshness in the large-scale mobile edge computing (MEC) networks. This work proposes the earliest analysis of the mean AoI (MAoI) performance of large-scale partial offloading MEC networks. Firstly, we derive and validate the closed-form expressions of MAoI by using queueing theory and stochastic geometry. Based on these expressions, we analyse the effects of computing offloading ratio (COR) and task generation rate (TGR) on the MAoI performance and compare the MAoI performance under the local computing, remote computing, and partial offloading schemes. The results show that by jointly optimising the COR and TGR, the partial offloading scheme outperforms the local and remote computing schemes in terms of the MAoI, which can be improved by up to 51\% and 61\%, respectively. This encourages the MEC networks to adopt the partial offloading scheme to improve the MAoI performance.
\end{abstract}


\section{Introduction}
\IEEEPARstart{W}{ith} 
the rapid development of Internet of Things (IoT) and artificial intelligence (AI), the freshness of information has become critical for latency-sensitive applications, such as autonomous driving, telemedicine, and augmented reality (AR). Decisions made by latency-sensitive applications based on outdated information may not be applicable. To characterise the information freshness, the age of information (AoI) has recently been proposed, which is defined as the time duration between the generation time of the latest computation-intensive task successfully received by the receiver and the current time \cite{ref1}.
Different from the traditional latency performance, the AoI is related to the task generation time at the source node and the queueing and processing time at the service node, which are irrelevant to the latency performance. However, user equipments (UEs) usually have limited computing capability, which increases the queueing and processing time, thus leading to an undesirable high AoI. In order to tackle the problem, the European Telecommunications Standard Institute (ETSI) proposes the concept of mobile edge computing (MEC) \cite{ref2}.
In the MEC networks, edge servers, such as base station (BS) or access point (AP), are usually deployed at the edge of networks with abundant computing resources. By offloading the computation-intensive tasks from UEs to edge servers, the queueing and processing time can be effectively reduced, thus improving the AoI performance \cite{ref3}. 
However, if all computation-intensive tasks are offloaded to edge servers, the AoI performance may not be improved due to limited spectrum resources, which increases the transmission delay. Therefore, how to determine the computing offloading scheme in the MEC networks has become a crucial problem to improve the AoI performance \cite{ref4}.

Some existing works studied the AoI performance in the MEC networks \cite{ref5,ref6,ref7,ref8,ref9}.
In \cite{ref5}, the mean AoI (MAoI), which is defined as the average value of AoI for the computation-intensive tasks in a given time interval, was investigated under both local and remote computing schemes. The results showed that the local computing scheme can achieve lower MAoI than the remote computing scheme when the size of tasks is large.
Based on this work, the partial offloading scheme has been proposed in \cite{ref6} to further improve the MAoI performance in the MEC networks. The closed-form results of MAoI under the partial offloading scheme were derived. The results showed that the MAoI under the partial offloading scheme can be significantly reduced compared to both local and remote computing schemes in \cite{ref5}, and there exists an optimal computing offloading ratio (COR), defined as the proportion of tasks offloaded to the edge server, to minimise the MAoI.
Moreover, the joint effects of COR and task generation rate (TGR), which is defined as the arrival rate of UE-generated tasks, on the MAoI performance in the partial offloading MEC networks were investigated in \cite{ref7,ref8,ref9}.
Works \cite{ref7} and \cite{ref8} compared the MAoI performance under the partial offloading scheme with joint optimisation of COR and TGR versus only COR optimisation. The results showed that the MAoI performance can be further improved by the joint optimisation of COR and TGR.
Work \cite{ref9} proposed a status update algorithm based on reinforcement learning to reduce the MAoI in the vehicular edge networks, where the roadside units are considered as edge servers, by jointly adjusting the COR and TGR. 
However, the tasks considered in the aforementioned works are code-partitioned rather than data-partitioned \cite{ref10}. Data-partitioned tasks can be divided into several independent sub-tasks for parallel computing, while code-partitioned tasks can only be computed in a specific sequence \cite{ref11}. With the growing interest in distributed machine learning, the AoI performance analysis with data-partitioned tasks has become critical, and existing models adopted in the aforementioned works cannot be applied. Moreover, all the above works only considered small-scale partial offloading MEC networks which consist of a limited number of edge servers and UEs. The effects of COR and TGR on the MAoI performance in large-scale partial offloading MEC networks are not yet known.

Some works studied the AoI performance in large-scale networks based on stochastic geometry \cite{ref12,ref13,ref14,ref15,ref16}. 
Work \cite{ref12} investigated the MAoI performance in the large-scale privacy preservation-oriented mobile crowd-sensing networks based on the Poisson point process (PPP). The results showed that the AoI for privacy-preserving data is higher compared to that for unencrypted sensing data. In \cite{ref13}, the MAoI performance was investigated in the wireless powered stochastic energy harvesting networks, where the power beacons and energy harvesting nodes follow the PPP and binomial point process, respectively. Its results showed that a lower MAoI can be achieved by optimising TGR when the successful transmission probability (STP), which is defined as the probability that the signal-to-interference ratio (SIR) of the destination node receiving the data packets from the source node is greater than a threshold, is relatively small. Work \cite{ref14} investigated the MAoI performance in the cellular-connected unmanned aerial vehicle (UAV) networks based on the PPP. The results showed that there exists an optimal UAV altitude for minimising the MAoI in the cellular-connected UAV networks.  
However, the clustered nature of UEs has been ignored. In our previous work \cite{ref15}, the MAoI performance of large-scale cellular IoT networks was analysed based on the Poisson cluster process, but this work focuses on the impact of data compression in the IoT devices on the MAoI, ignoring the effect of computing offloading strategy on the MAoI in the MEC network. To our best knowledge, work \cite{ref16} is the only work investigating the MAoI performance in the large-scale vehicular networks with computing offloading. However, it only studied the MAoI performance under the remote offloading strategy, i.e., offloading all computation-intensive tasks to the edge server, where the impact of partial offloading strategy is missing. Furthermore, there lacks of literatures analysing the joint effects of COR and TGR on the MAoI performance in the partial offloading MEC networks.  
Accordingly, the motivations of our work can be summarised as follows:
\begin{enumerate}
\item The AoI performance analysis in the large-scale partial offloading MEC networks is missing in the literature. The partial offloading scheme has only been investigated in the small-scale MEC networks, which consist of a limited number of edge servers and UEs. How this scheme performs in the large-scale MEC networks is still unknown.
\item The existing queueing models may no longer be applicable to the large-scale partial offloading MEC networks. 
\item There lacks of literatures investigating the joint effects of COR and TGR in the large-scale partial offloading MEC networks.
\end{enumerate}

To our best knowledge, this paper is the earliest analysis on the MAoI performance in the large-scale partial offloading MEC networks. Based on the Jackson network in queueing theory, we propose a new queueing model for the partial offloading scheme with data-partitioned tasks. Equipped with these, the closed-form expressions of MAoI in the large-scale MEC networks under the local computing, remote computing, and partial offloading schemes are derived and validated by Monte Carlo simulations. Then, we numerically analyse the effects of COR and TGR on the MAoI performance and compare the MAoI performance under local computing, remote computing, and partial offloading schemes. Our main contributions are summarised as follows:
\begin{enumerate}
	\item We proposes the earliest analysis of the MAoI performance in the large-scale partial offloading MEC networks. By combining queueing theory and stochastic geometry, the closed-form expressions of MAoI in the large-scale partial offloading MEC networks are derived and validated through Monte Carlo simulations.
	\item To obtain the closed-form expressions of MAoI in the large-scale partial offloading MEC networks, we adopt the Jackson network in queueing theory to describe the partial offloading scheme with data-partitioned tasks.
	\item The results show that by jointly optimising COR and TGR, the MAoI can be significantly reduced. Under our simulation environment, the MAoI under the partial offloading scheme with the optimal COR and TGR can be reduced by up to 51\% and 61\% compared to the local and remote computing schemes, respectively. Consequently, the partial offloading scheme can be widely adopted in the large-scale MEC networks to improve the MAoI performance. Additionally, the proposed partial offloading scheme can alleviate the requirement on the computing capability of UEs.
\end{enumerate}

The remainder of this article is organised as follows. Section \ref{sec_two} introduces the network model, the computing offloading model, and the investigated performance metrics. In Section \ref{sec_three}, the analytical results of MAoI in the large-scale MEC networks are derived. In Section \ref{sec_four}, the validation results and the numerical analysis are presented before concluding the paper in Section \ref{sec_five}.

\section{System Model}\label{sec_two}
\subsection{Network Model}
We consider a large-scale partial offloading MEC network consisting of UEs and BSs. Each UE consists of three components: a sensor that generates the computation-intensive tasks, a local processor that processes these tasks, and a transmitter that offloads these tasks to its associated BS via wireless channels. Each BS is equipped with an edge server that provides edge computing resources to process the offloaded tasks and then transmits the computing results back to the corresponding UE for decision-making. 
The positions of BSs and UEs are modelled following a Matern cluster process (MCP) \cite{ref17}. The locations of BSs follow a PPP with density $\lambda_{\rm{B}}$, denoted by $\{\boldsymbol{{\rm{y}}}_i\}$, $i \in (1,2,\cdots,N_{\rm{B}})$ with $N_{\rm{B}}$ being the number of BSs in the partial offloading MEC network. The locations of UEs are uniformly scattered in a circular region centered around their associated BSs with radius being $1/\sqrt{\pi\lambda_{\rm{B}}}$ \cite{ref18}. The locations of UEs are represented by $\{\boldsymbol{{\rm{x}}}_{i,j}\}$, $j \in (1,2,\cdots,N_i)$, where $N_i$ is the number of UEs in the BS $\boldsymbol{{\rm{y}}}_{i}$. 

For each link between the UE and its associated BS, we assume that it experiences pathloss and small-scale fading. 
For the pathloss, a log-distance model is assumed, which is expressed by $l(r)=r^{-\alpha}$ \cite{ref19}, where $r$ is the distance between the UE and its associated BS, and $\alpha$ is the pathloss exponent. 
For the small-scale fading, Rayleigh fading is assumed. Therefore, the received power attenuation can be modelled as an independent exponential distribution with rate parameter being 1. 
We further assume that the uplink transmit power employs a fraction channel inversion power control mechanism \cite{ref20}. Consequently, for a specific UE $\boldsymbol{{\rm{x}}}_{i,j}$ associated with the BS $\boldsymbol{{\rm{y}}}_{i}$, the uplink transmit power $P_{i,j}$ can be expressed by $P_{i,j}=pl(R_{i,j})^{-\epsilon}$, where $p$ is the fixed transmit power, $R_{i,j}$ is the distance between the UE $\boldsymbol{{\rm{x}}}_{i,j}$ and its associated BS $\boldsymbol{{\rm{y}}}_{i}$, and $\epsilon$ is the power control factor. 

According to Slivnyak's theorem \cite{ref21}, we deploy a typical BS at the origin without loss of generality. We assume that a group of UEs served by the same BS uses orthogonal frequency resources \cite{ref22}. Therefore, for the typical BS, there is only an interfering UE in each BS. 
For denotational simplicity, $\{R_i\}$ and $\{D_i\}$ denote the distances from the interfering UEs to their associated BSs and the typical BS, respectively. 
In particular, the distance between a specific UE associated with the typical BS and the typical BS is denoted by $R_0$. 
As a result, the received SIR at the typical BS can be expressed by 
\begin{equation}\label{eqn-1} 
	{\rm{SIR}} = \frac{{{h_0}{R_0}^{ - \alpha }pl{{({R_0})}^{ -\epsilon }}}}{{\sum\nolimits_{\{\Phi_i\}} {{h_i}{D_i}{^{ - \alpha }}pl{{({R_i})}^{ -\epsilon }}} }},
\end{equation}
where $\Phi_i$ is the set of interfering UEs. $h_0$ and $h_i$ are the small-scale fading from the specific UE and the $i$-th interfering UE to the typical BS, respectively. Since this work focuses on the MAoI performance in the large-scale partial offloading MEC networks, it is assumed that the BSs can obtain perfect channel state information (CSI) via channel estimation on pilot signals from the UEs \cite{ref23}. 

Equipped with the SIR, if a computation-intensive task is successfully offloaded from the UE $\boldsymbol{{\rm{x}}}_{i,j}$ to its associated BS $\boldsymbol{{\rm{y}}}_{i}$, the average offloading rate $\bar{C}^{\rm{tr}}_{i,j}$ can be expressed by $B_{i,j}\log_2(1+\tau)\Theta_{i,j}$ \cite{ref24}, where $B_{i,j}$ is the bandwidth of sub-channel. In our work, the round-robin scheduling is adopted to fairly allocate each sub-channel bandwidth \cite{ref25}. 
Thus, the sub-channels bandwidth $B_{i,j}$ can be expressed by $B_{i,j}=B_{{\rm{tot}}}/N_i$, where $B_{{\rm{tot}}}$ denotes the available total bandwidth in the BS $\boldsymbol{{\rm{y}}}_{i}$. Additionally, $\tau$ is a SIR threshold. $\Theta_{i,j}$ denotes the STP, which is defined as the probability that the SIR of BS $\boldsymbol{{\rm{y}}}_{i}$ receiving the computation-intensive tasks from the UE $\boldsymbol{{\rm{x}}}_{i,j}$ exceeds the SIR threshold $\tau$. 

\subsection{Computing Offloading Model} 
\begin{figure}[!t]
	\centering
	\includegraphics[width=2.9in]{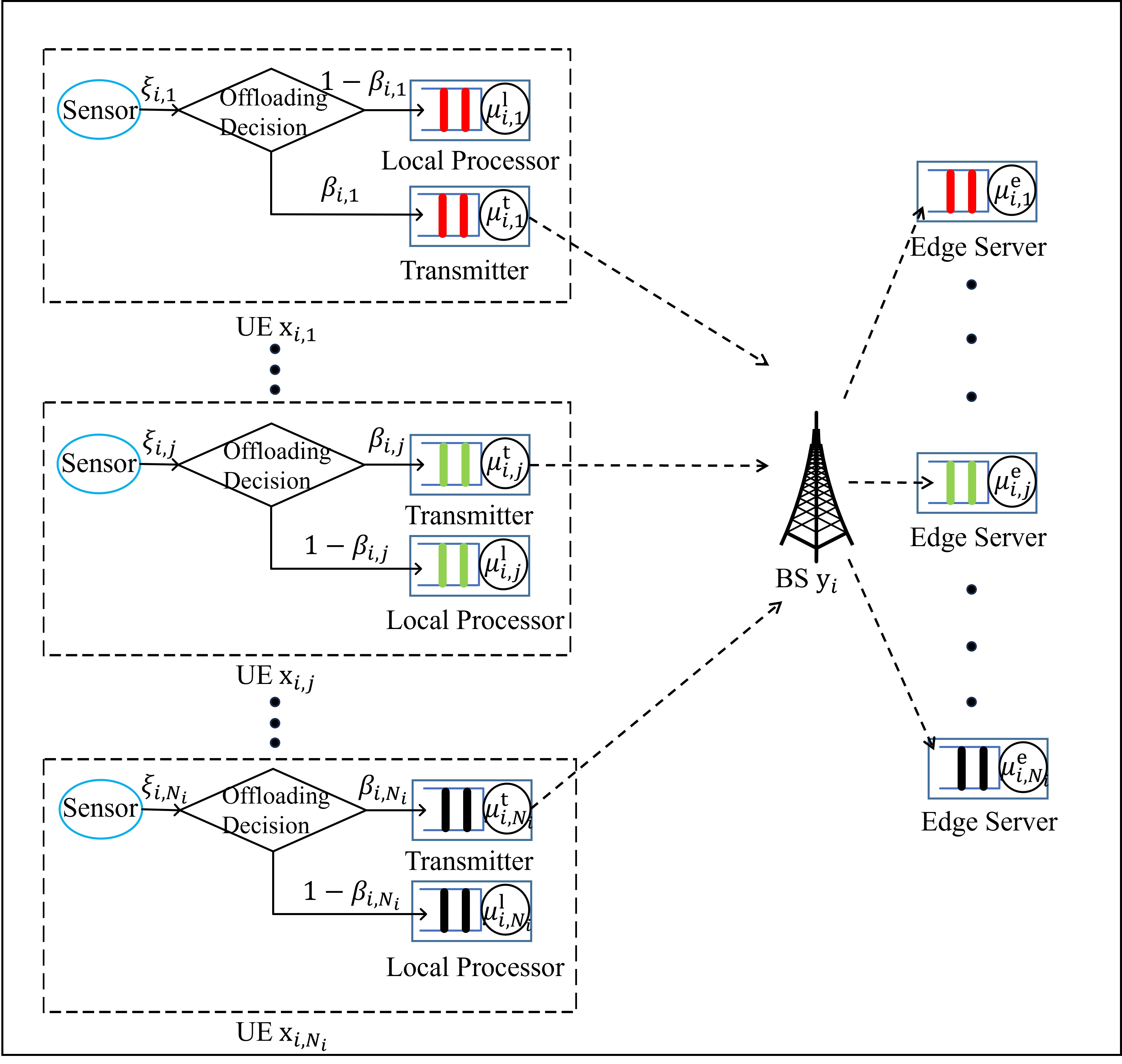}
	\caption{The computing offloading model in the large-scale partial offloading MEC network.}
	\label{fig1}
	\vspace{-5pt}
\end{figure}
As shown in Fig. \ref{fig1}, the entire procedure of computing offloading is described as follows. 
Firstly, the computation-intensive tasks are generated at each sensor of UE. For a specific UE $\boldsymbol{{\rm{x}}}_{i,j}$ associated with the BS $\boldsymbol{{\rm{y}}}_i$, we assume that the size of tasks follows an independently and identically distributed exponential distribution with a fixed mean $L_{i,j}$. Meanwhile, the task generation is assumed to be event-triggered \cite{ref26}, i.e., the sensor of the UE detects the occurrence of a certain event and then generates a computation-intensive task. The generated task is often random. Therefore, we model the task arrival process as a Poisson process with rate parameter $\xi_{i,j}$ \cite{ref27}. 
Secondly, based on the computing capability of the UE, the UE determines whether this task can be offloaded and whether to fully or partially offload the task to its associated BS. Since information about the computing capability of UEs can be easily obtained, we assume that the time consumed in this decision process can be ignored \cite{ref19}.

In this work, we consider three computing offloading schemes based on the proportion of computation-intensive tasks being processed at the UE, namely the local computing, remote computing, and partial offloading schemes. For the local computing scheme, all computation-intensive tasks are processed by the local processor of the UE. For the remote computing scheme, all computation-intensive tasks are first offloaded to the BS and then processed by the edge server of the BS. For the partial offloading scheme, some of the computation-intensive tasks are processed by the local processor of the UE, while the rest of them are processed by the edge server of the BS. 
We assume that the computation-intensive tasks are offloaded and processed following a first-come-first-served (FCFS) queueing discipline and that the buffers of the transmitters, local processors, and edge servers have infinite sizes. Additionally, it is assumed that the unsuccessful offloaded tasks will not be re-offloaded \cite{ref28}.

Three computing offloading schemes are described in details as follows.
\subsubsection{Local Computing Scheme}
In this scheme, all computation-intensive tasks are processed by the local processor of the UE $\boldsymbol{{\rm{x}}}_{i,j}$. Recall that the task arrival process follows the Poisson process with rate $\xi_{i,j}$ and the task size follows an exponential distribution with mean $L_{i,j}$. Thus, the local computing process can be modelled as the M/M/1 queue with the arrival rate being $\xi_{i,j}$ and a constant service rate being $\kappa_{i,j}$. Meanwhile, the average service delay $G_{i,j}$ at the local processor of the UE $\boldsymbol{{\rm{x}}}_{i,j}$ can be expressed by
\begin{equation}\label{eqn-2}
	G_{i,j}=\frac{C_{i,j}L_{i,j}}{f_{i,j}},
\end{equation} 
where $C_{i,j}$ is the required CPU cycles to process one bit data, and $f_{i,j}$ is the computing frequency at local processor of UE $\boldsymbol{{\rm{x}}}_{i,j}$, indicating the number of CPU cycles that can be provided per second. Based on the Utilization law \cite{ref29}, the service rate is inversely proportional to the average service time for a single server. Accordingly, the service rate $\kappa_{i,j}$ at the local processor of the UE $\boldsymbol{{\rm{x}}}_{i,j}$ can be expressed by $\kappa_{i,j}=1/{G_{i,j}}$.

\subsubsection{Remote Computing Scheme}
In this scheme, the computation-intensive tasks generated by the UE $\boldsymbol{{\rm{x}}}_{i,j}$ are offloaded to its associated BS $\boldsymbol{{\rm{y}}}_{i}$ by the transmitter and then processed by the edge server of the BS $\boldsymbol{{\rm{y}}}_{i}$. 
Similar to the local computing process, the offloading process can be modelled as the M/M/1 queue with the arrival rate being $\xi_{i,j}$ and the service rate being $\upsilon_{i,j}$. 
The average offloading delay $K_{i,j}$ at the transmitter of the UE $\boldsymbol{{\rm{x}}}_{i,j}$ can be expressed by
\begin{equation}\label{eqn-3}
	K_{i,j} = \frac{{{L_{i,j}}}}{{{\bar{C}_{i,j}^{\rm{tr}}}}}= \frac{{{L_{i,j}}}}{{{B_{i,j}\log_2(1+\tau)\Theta_{i,j}}}}.
\end{equation}
Then, the service rate $\upsilon_{i,j}$ at the transmitter of the UE $\boldsymbol{{\rm{x}}}_{i,j}$ can be obtained by $1/K_{i,j}$.

According to Burke's theory \cite{ref30}, the output of the M/M/1 queue is still a Poisson process with the same arrival rate. Therefore, we model the computing process of the edge server as the M/M/1 queue with the arrival rate $\xi_{i,j}$ and the service rate being $\varepsilon_{i,j}$. Moreover, we assume that the computing frequency of the BS $\boldsymbol{{\rm{y}}}_{i}$ are divided into multiple independent units as virtual machines (VMs) through network function virtualisation (NFV) technology \cite{ref22},\cite{ref31}. Each VM only processes the computation-intensive tasks from the corresponding UE. As a result, the sum of the computing frequency $g_{i,j}$ allocated to each UE in the BS $\boldsymbol{{\rm{y}}}_{i}$ cannot exceed the computing frequency $f^{\rm{B}}_{i}$ of the BS $\boldsymbol{{\rm{y}}}_{i}$, i.e., $\sum\limits_{j = 1}^{j={N_i}} {g_{i,j}} \le f_i^{\rm{B}}$. Consequently, the average service delay $H_{i,j}$ at the edge server can be expressed by
\begin{equation}\label{eqn-4}
	H_{i,j} = \frac{{{C_{i,j}}{L_{i,j}}}}{{g_{i,j}}}.
\end{equation}
Then, the service rate $\varepsilon_{i,j}$ at the edge server can be obtained by $1/H_{i,j}$.

\subsubsection{Partial Offloading Scheme}
In this scheme, the computation-intensive tasks generated by the UE $\boldsymbol{{\rm{x}}}_{i,j}$ are divided into two independent parts, i.e., $\beta_{i,j}$ and $(1-\beta_{i,j})$, where $\beta_{i,j}\in(0,1)$ is the COR. As illustrated in Fig. \ref{fig1}, a portion $(1-\beta_{i,j})$ of computation-intensive tasks is processed by the UE $\boldsymbol{{\rm{x}}}_{i,j}$, while the rest $\beta_{i,j}$ of them are processed by the BS $\boldsymbol{{\rm{y}}}_{i}$ for parallel computing. Based on queueing theory, we model the partial offloading process as a Jackson network, which consists of one M/M/1 queue described in the local computing scheme and two M/M/1 tandem queues described in the remote computing scheme. According to the independent thinning property of Poisson process, the arrivals of computation-intensive tasks at the local processor and transmitter of the UE $\boldsymbol{{\rm{x}}}_{i,j}$ follow the Poisson processes with rates $(1-\beta_{i,j})\xi_{i,j}$ and $\beta_{i,j}\xi_{i,j}$, respectively. Accordingly, for local computing in the partial offloading scheme, the average service delay $S_{i,j}^{\rm{l}}$ at the local processor of the UE $\boldsymbol{{\rm{x}}}_{i,j}$ can be expressed by $S_{i,j}^{\rm{l}}=(1-\beta_{i,j})G_{i,j}$. 
Consequently, the service rates $\mu_{i,j}^{\rm{l}}$ of local processor can be calculated by $1/S_{i,j}^{\rm{l}}$. For remote computing in the partial offloading scheme, the average service delay $S_{i,j}^{\rm{t,e}}$ consists of the average offloading delay $S_{i,j}^{\rm{t}}$ at the transmitter of UE $\boldsymbol{{\rm{x}}}_{i,j}$ and the average service delay $S_{i,j}^{\rm{e}}$ at the edge server of BS $\boldsymbol{{\rm{y}}}_{i}$, which can be expressed by $S_{i,j}^{\rm{t}}=\beta_{i,j} K_{i,j}$ and $S_{i,j}^{\rm{e}}=\beta_{i,j}H_{i,j}$, respectively. As a result, under the partial offloading scheme, the service rate $\mu_{i,j}^{\rm{t}}$ of transmitter and the service rate $\mu_{i,j}^{\rm{e}}$ of edge server can be calculated by $1/S_{i,j}^{\rm{t}}$ and $1/S_{i,j}^{\rm{e}}$, respectively.  

Note that the local and remote computing schemes can be considered as special cases of the partial offloading scheme. In the partial offloading scheme, $\beta_{i,j}=0$ corresponds to the local computing scheme, while $\beta_{i,j}=1$ corresponds to the remote computing scheme.

After the computing offloading process, the computing results at the BS are transmitted back to the corresponding UE for decision-making. Since the size of computing results is limited, the delay in transmitting the computing results can be ignored in our work \cite{ref6}.

\subsection{Investigated Performance Metrics}
In this work, we focus on investigating the AoI performance in the large-scale partial offloading MEC networks. Recall that the AoI is defined as the time elapsed between the generation time of the latest completely processed computation-intensive task received by the UE $\boldsymbol{{\rm{x}}}_{i,j}$ and the current time. Thus, the AoI $\Delta_{i,j}(t)$ at the UE $\boldsymbol{{\rm{x}}}_{i,j}$ can be expressed by $\Delta_{i,j}(t) = t - t_{z}$ \cite{ref1}, where $t$ denotes the current time. $t_{z}$ denotes the generation time of the most recent completely processed computation-intensive task index $z$ received by the UE $\boldsymbol{{\rm{x}}}_{i,j}$. However, as the AoI is a stochastic process, it is difficult to obtain the closed-form expression of $\Delta_{i,j}(t)$. Therefore, we adopt the MAoI, which is defined as the temporal average AoI in a given time interval, as the performance metrics to evaluate the information freshness in the large-scale partial offloading MEC networks. As a result, the MAoI $\bar\Delta _{i,j}$ of the UE $\boldsymbol{{\rm{x}}}_{i,j}$ in the time interval $(0,\mathcal{L})$ can be expressed by
\begin{equation}\label{eqn-7}
	\bar \Delta _{i,j}  = \mathop {\lim }\limits_{\mathcal{L}  \to \infty }\frac{1}{\mathcal{L}} \int_0^\mathcal{L}  {\Delta_{i,j} (t){\rm{d}}t}.
\end{equation}

The derivations of MAoI in the large-scale partial offloading MEC networks will be provided in the next section.

\section{Performance Analysis}\label{sec_three}
In this section, the STP of the computation-intensive tasks in the large-scale partial offloading MEC networks will be given first. Then, the closed-form expressions of MAoI under the local computing, remote computing, and partial offloading schemes will be derived. Based on these expressions, the MAoI performance in the large-scale partial offloading MEC networks can be analysed.

\subsection{Successful Transmission Probability}
The STP $\Theta_{i,j}$ is defined as the probability that the SIR of the link from the UE $\boldsymbol{{\rm{x}}}_{i,j}$ to the BS $\boldsymbol{{\rm{y}}}_{i}$ is greater than a SIR threshold $\tau$. Therefore, the STP $\Theta_{i,j}$ can be expressed by
\begin{equation}\label{theta}
	\Theta_{i,j}={\rm{P}}\left({\rm{SIR}}>\tau\right).
\end{equation}

Based on our previous work \cite{ref24}, the closed-form result of the STP in large-scale clustered networks modelled by the MCP has been investigated. Since the spatial distributions of BSs and UEs in our work also follow the MCP, the STP in the large-scale partial offloading MEC networks can be obtained similarly, which is given in Lemma \ref{lemmastp}. 
\begin{mylemma}\label{lemmastp}
In the large-scale partial offloading MEC networks with the spatial distributions of UEs and BSs following the MCP, the STP $\Theta_{i,j}$ can be expressed by
	\begin{equation}\label{eqn-stp}
		{\Theta _{i,j}} = \left\{ \begin{aligned}
			&{\bf{E}}_{\frac{\epsilon}{\epsilon-1}}\left( \varsigma  \right) + {\varsigma}^{\frac{1}{{1 - \epsilon}}} \Gamma \left( {1 + \frac{1}{{1 - \epsilon}}} \right),& 0 &\le  \epsilon < 1\\
			&\exp \left( { - \varsigma } \right),& \epsilon&= 1,
		\end{aligned} \right.
	\end{equation}
where $\varsigma=\frac{2\pi\tau^{\frac{2}{\alpha}}}{\alpha(1+\epsilon)\sin\left(\frac{2\pi}{\alpha}\right)}$ and ${\bf{E}}_{\frac{\epsilon}{\epsilon-1}}\left( \varsigma  \right)=\int_1^\infty \frac{e^{ - \varsigma t}}{t^{\frac{\epsilon}{\epsilon-1}}} {\rm{d}}t$.
\end{mylemma}
\begin{proof}
	The derivation of STP is detailed in Appendix B of \cite{ref24}.
\end{proof}

Based on the result given in Lemma \ref{lemmastp}, the close-form expressions of MAoI under the local computing, remote computing, and partial offloading schemes can be derived.

\subsection{Age of Information Analysis} 
\begin{figure}[t]
	\centering
	\includegraphics[width=3in]{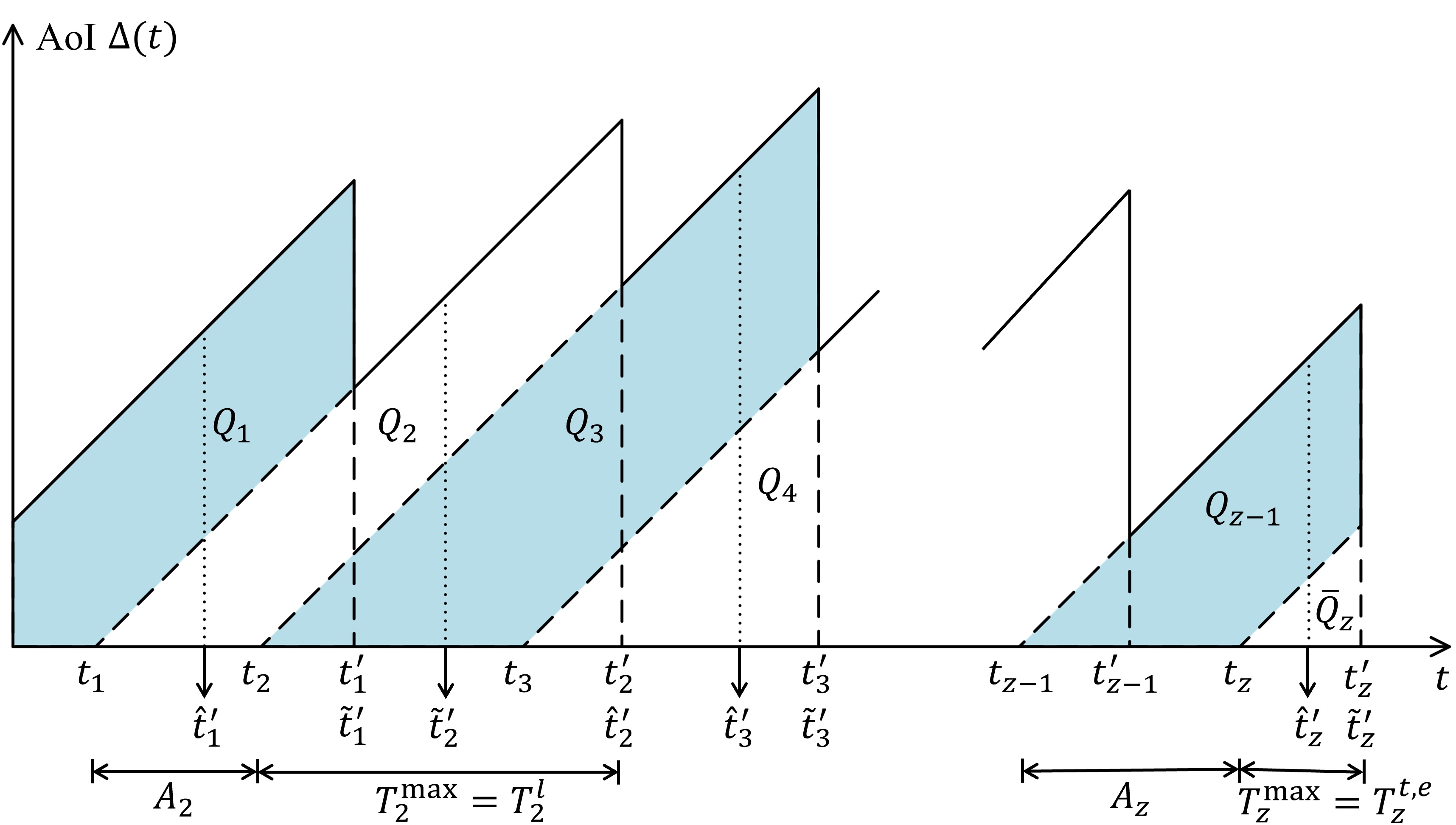}
	\caption{An example of AoI demonstration under the partial offloading scheme with the FCFS queueing discipline.}
	\label{fig2}
\end{figure}
For the partial offloading scheme, the UE $\boldsymbol{{\rm{x}}}_{i,j}$ successfully receives a complete computing result of computation-intensive task, which is obtained by integrating the computing results from the UE $\boldsymbol{{\rm{x}}}_{i,j}$ and the BS $\boldsymbol{{\rm{y}}}_{i}$. Then, the AoI decreases to the elapsed time between the generation of computation-intensive task and the reception of its complete computing result. Otherwise, the AoI increases linearly. The AoI evolution under the partial offloading scheme can be shown in Fig. \ref{fig2}. The $n$-th computation-intensive task is generated at time $t_n$ for $n\in(1,\cdots,z)$, where $z$ is the most recent index of task in the time interval $(0,\mathcal{L})$. Since the $n$-th task is divided by the UE into two independent parts, one part of the task is processed by the local processor of the UE and completed at time ${\widehat{t}}_{n}^{\prime}$, while the rest of the task is processed by the edge server of the BS and completed at the time ${\widetilde{t}}_{n}^{\prime}$. Therefore, the UE obtains a complete computing result at the time $t_n^{\prime}$, where $t_n^{\prime}=\max\{{\widehat{t}}_{n}^{\prime},{\widetilde{t}}_{n}^{\prime}\}$. According to \eqref{eqn-7}, the MAoI of the UE $\boldsymbol{{\rm{x}}}_{i,j}$ can be obtained by calculating the area under the AoI curve in a certain time interval $(0, \mathcal{L})$, where $\mathcal{L}$ is set to $t_z^{\prime}$. 
As illustrated in Fig. \ref{fig2}, this area can be considered as consisting of a polygon $Q_1$, multiple isosceles trapezoids $Q_n, n \in \{2,3,\cdots,z-1\}$, and an isosceles triangle $\bar{Q}_z$.
As a result, the MAoI $\bar \Delta _{i,j}^{{\rm{p}}}$ under the partial offloading scheme can be expressed by 
\begin{equation}\label{eqn-8}
	\bar \Delta _{i,j}^{{\rm{p}}} = \mathop {\lim }\limits_{{\cal L} \to \infty } \frac{1}{{\cal L}}\left({Q_1} + \sum_{n = 2}^{n=z - 1} {{Q_n}}  + {{\bar Q}_z}\right),
\end{equation}
which can be calculated by using the result in Lemma \ref{lemma1}.
\begin{mylemma}\label{lemma1}
In the large-scale partial offloading MEC networks, the MAoI $\bar \Delta _{i,j}^{{\rm{p}}}$ of the UE $\boldsymbol{{\rm{x}}}_{i,j}$ under the partial offloading scheme with the FCFS queueing discipline can be calculated by
	\begin{equation}\label{eqn-9}
		\begin{aligned}
			\bar \Delta _{i,j}^{{\rm{p}}}& = \frac{1}{{\mathbb{E}\left[ {{A_n}} \right]}}\bigg( { \left( {1 - {\beta _{i,j}}} \right){\rm{P}}\left( {E_n^{\rm{l}}} \right)\mathbb{E}\left[ {A_n^{\rm{l}}T_n^{\rm{l}}\left| E_n^{\rm{l}} \right.} \right]}\\
			& +{{\beta _{i,j}}{\rm{P}}\left( {E_n^{\rm{t,e}}} \right)\mathbb{E}\left[ {A_n^{\rm{t}} {T_n^{\rm{t,e}}}\left| E_n^{\rm{t,e}} \right.} \right]} +\frac{1}{2}\mathbb{E}\left[ {{A_n}^2} \right] \!\!\bigg),
		\end{aligned}
	\end{equation}
where $A_n$, $A_n^{\rm{l}}$, and $A_n^{\rm{t}}$ denote the arrival intervals between the $(n-1)$-th and $n$-th computation-intensive tasks at the sensor, the local processor, and the transmitter of the UE, respectively. $T^{\rm{l}}_n$ and $T_n^{\rm{t,e}}$ denote the system times of the $n$-th computation-intensive task for local computing and remote computing, respectively. $E_n^{\rm{l}}$ denotes the event that the system time of the $n$-th computation-intensive task for local computing exceeds that for remote computing, i.e., $T_n^{\rm{l}}>T_n^{\rm{t,e}}$. $E_n^{\rm{t,e}}$ denotes the event that the system time of the $n$-th computation-intensive task for remote computing exceeds that for local computing, i.e., $T_n^{\rm{t,e}}>T_n^{\rm{l}}$. 
\end{mylemma}
\begin{myproof}
	See Appendix \ref{lp1}.
\end{myproof}

Equipped with the result, the closed-form expression of MAoI under the partial offloading scheme is given in Theorem \ref{theorem1}.
\begin{mytheorem}\label{theorem1}
In the large-scale partial offloading MEC networks, the closed-form expression of MAoI at the UE $\boldsymbol{{\rm{x}}}_{i,j}$ under the partial offloading scheme with the FCFS policy can be calculated by
\begin{equation}\label{eqn-14}
\begin{aligned}
\bar \Delta _{i,j}^{{\rm{p}}} &= \frac{{{\xi _{i,j}}}}{{\omega _{i,j}^{{\rm{t}}}\omega _{i,j}^{{\rm{e}}}}}\bigg[ {\left( {\mu _{i,j}^{\rm{t}} - {\beta _{i,j}}{\xi_{i,j}}} \right)\left( {\mu _{i,j}^{\rm{e}} - {\beta _{i,j}}{\xi _{i,j}}} \right)\Xi _{i,j}^{1}}\\
& + \chi^{\rm{l}}_{i,j}\left( {{\gamma _{i,j}} + {\beta _{i,j}}{\xi _{i,j}}} \right)\bigg(\Xi _{i,j}^{2}+\Xi_{i,j}^{3} \\
& +  {\frac{{{\beta _{i,j}^2}{\xi _{i,j}}}}{{\Omega _{i,j}^{\rm{t,e} }}}\Xi _{i,j}^{4} + \frac{{\beta _{i,j}^2\xi_{i,j}\chi _{i,j}^{\rm{t}}}}{{\mu _{i,j}^{\rm{e}}\Omega _{i,j}^{\rm{t,e} }}}\Xi _{i,j}^{5}\bigg)} \bigg] + \frac{1}{{{\xi _{i,j}}}},
\end{aligned}
\end{equation}
where $\omega_{i,j}^{{\rm{\psi}}}\!\!=\!\!\mu_{i,j}^{\rm{l}}+\mu_{i,j}^{\rm{\psi}}-\xi_{i,j}$, $\chi^{\rm{l}}_{i,j}\!\!=\!\!\mu_{i,j}^{\rm{l}}\!-\!\left(1-\beta_{i,j}\right)\xi_{i,j}$,  $\gamma_{i,j}=\mu_{i,j}^{\rm{l}}+\mu_{i,j}^{\rm{t}}+\mu_{i,j}^{\rm{e}}-\xi_{i,j}$, $\Omega_{i,j}^{{\rm{\nu,\psi}}}=\mu_{i,j}^{\rm{\nu}}+\mu_{i,j}^{\rm{\psi}}-\beta_{i,j}\xi_{i,j}$, and $\chi^{\rm{t}}_{i,j}\!=\!\mu_{i,j}^{\rm{t}}\!-\!\beta_{i,j}\xi_{i,j}$ are defined for denotational simplicity. In these parameters, $\psi \in \{\rm{t,e}\}$ and $\nu\in\{\rm{l,t,e}\}$.
Additionally, $\Xi_{i,j}^{{1}}$, $\Xi_{i,j}^{{2}}$, $\Xi_{i,j}^{{3}}$, $\Xi_{i,j}^{{4}}$, and $\Xi_{i,j}^{{5}}$ are given at the bottom of the page. 
\end{mytheorem}
\begin{myproof}
	See Appendix \ref{proof3}.
\end{myproof}
\begin{figure*}[b]
	\hrule
	\vspace{0.2cm}
	\begin{equation}\label{eqn-ATL}
		\begin{aligned}
			\Xi _{i,j}^{1} &= \frac{1}{{\mu _{i,j}^{\rm{l}}{\xi _{i,j}}}} + \frac{{\chi_{i,j}^{\rm{l}}\left( {\Omega _{i,j}^{\rm{l,t} } + \omega _{i,j}^{\rm{e}}} \right)}}{{\mu _{i,j}^{\rm{l}}{\xi _{i,j}}\Omega _{i,j}^{\rm{l,t} }\Omega _{i,j}^{\rm{l,e} }}} + \frac{{\left( {1 - {\beta _{i,j}}} \right)^2\xi_{i,j}}}{\mu_{i,j}^{\rm{l}}{\Omega _{i,j}^{\rm{l,t} }\Omega _{i,j}^{\rm{l,e} }}}\left[ {\frac{{\Omega _{i,j}^{\rm{l,t} } + \Omega _{i,j}^{\rm{l,e} } - \mu _{i,j}^{\rm{l}}}}{{\Omega _{i,j}^{\rm{l,e} }\omega _{i,j}^{\rm{e}} - \Omega _{i,j}^{\rm{l,t} }\omega _{i,j}^{\rm{t}}}}\left( {\frac{{\Omega _{i,j}^{\rm{l,e} }\omega _{i,j}^{\rm{e}}}}{{\omega _{i,j}^{\rm{t}}}} - \frac{{\Omega _{i,j}^{\rm{l,t} }\omega _{i,j}^{\rm{t}}}}{{\omega _{i,j}^{\rm{e}}}}} \right)} \right. \\
			&+ \left. {\frac{{{3\mu _{i,j}^{\rm{t}} - 2{\beta _{i,j}}{\xi _{i,j}}}}}{{\chi_{i,j}^{\rm{l}}}}\left( {\frac{{\left( {\mu _{i,j}^{\rm{t}} - {\beta _{i,j}}{\xi _{i,j}}} \right)\left( {\mu _{i,j}^{\rm{e}} - {\beta _{i,j}}{\xi _{i,j}}} \right)}}{{\mu _{i,j}^{\rm{l}}}} + \Omega _{i,j}^{\rm{l,t} } + \Omega _{i,j}^{\rm{l,e} } - \mu _{i,j}^{\rm{l}}} \right)} \right]-\frac{{2\left( {1 - {\beta _{i,j}}} \right)^2}}{{{{\left( {\mu _{i,j}^{\rm{l}}} \right)}^3}}}.
		\end{aligned}
	\end{equation}
	\hrule
	\vspace{0.2cm}
	\begin{equation}\label{eqn-ATT}
		\begin{aligned}
			\Xi_{i,j}^{2}&=\frac{1}{{\mu _{i,j}^{\rm{t}}{\xi _{i,j}}}} + \frac{{\left( {\mu _{i,j}^{\rm{t}} - {\beta _{i,j}}{\xi _{i,j}}} \right)\left( {\mu _{i,j}^{\rm{e}} - {\beta _{i,j}}{\xi _{i,j}}} \right)}}{{\mu _{i,j}^{\rm{t}}{\xi _{i,j}}\Omega_{i,j}^{\rm{l,t}}\left( {\Omega _{i,j}^{\rm{l,e}} + \omega _{i,j}^{\rm{t}} - \mu _{i,j}^{\rm{l}}} \right)}}+\frac{\beta _{i,j}^2\xi_{i,j}}{\mu _{i,j}^{\rm{t}}}\left[ {\frac{{ {3\mu _{i,j}^{\rm{t}} - 2{\beta _{i,j}}{\xi _{i,j}}} }}{{\mu _{i,j}^{\rm{t}}\left( {\mu _{i,j}^{\rm{t}} - {\beta _{i,j}}{\xi _{i,j}}} \right)}} } \right.+\left. {\frac{{ {\mu _{i,j}^{\rm{e}} - {\beta _{i,j}}{\xi _{i,j}}} }}{{\omega _{i,j}^{\rm{t}}\omega _{i,j}^{\rm{e}}\Omega_{i,j}^{\rm{l,t}}}}} \right]-\frac{2\beta_{i,j}^2\xi_{i,j}}{(\mu_{i,j}^{\rm{t}})^3}.
		\end{aligned}
	\end{equation}
	\hrule
	\vspace{0.2cm}
	\begin{equation}\label{eqn-ASR}
		\begin{aligned}
			\Xi_{i,j}^{3}=\frac{1}{{\mu _{i,j}^{\rm{e}}{\xi _{i,j}}}} + \frac{{\mu _{i,j}^{\rm{t}} - {\beta _{i,j}}{\xi _{i,j}}}}{{{\xi _{i,j}}\Omega_{i,j}^{\rm{l,e}}\gamma_{i,j}}}.
		\end{aligned}
	\end{equation}
	\hrule
	\vspace{0.2cm}
	\begin{equation}\label{eqn-AWR1}
		\begin{aligned}
			\Xi_{i,j}^{4} &= \frac{{\mu _{i,j}^{\rm{e}}\left(\mu_{i,j}^{\rm{t}}-\beta_{i,j}\xi_{i,j}\right)}}{\mu_{i,j}^{\rm{t}}{\omega _{i,j}^{\rm{t}}\Omega_{i,j}^{\rm{l,e}}}}\left( {\frac{1}{{\mu_{i,j}^{\rm{e}}-\beta_{i,j}\xi_{i,j}}} + \frac{1}{{\omega _{i,j}^{\rm{e}}}} + \frac{1}{\gamma_{i,j}}} \right)- \frac{1}{\mu_{i,j}^{\rm{t}}{\Omega _{i,j}^{\rm{t,e} }}} + \frac{{\chi_{i,j}^{\rm{l}}\gamma_{i,j}\left( {\mu _{i,j}^{\rm{t}} + 2\mu _{i,j}^{\rm{e}} - 2{\beta _{i,j}}{\xi _{i,j}}} \right)}}{\mu_{i,j}^{\rm{t}}{\omega _{i,j}^{\rm{t}}\left(\mu_{i,j}^{\rm{e}}-\beta_{i,j}\xi_{i,j}\right)\Omega _{i,j}^{\rm{t,e} }\Omega_{i,j}^{\rm{l,e}}}}.
		\end{aligned}
	\end{equation}
	\hrule
	\vspace{0.2cm}
	\begin{equation}\label{eqn-AWR2}
		\begin{aligned}
			\Xi_{i,j}^{5}&=\frac{\mu_{i,j}^{\rm{t}}+\mu_{i,j}^{\rm{e}}}{\mu_{i,j}^{\rm{t}}\mu_{i,j}^{\rm{e}}}\left[\frac{1}{{\mu_{i,j}^{\rm{e}}-\beta_{i,j}\xi_{i,j}}} - \frac{1}{{\Omega _{i,j}^{\rm{t,e} }}}+\frac{{\mu _{i,j}^{\rm{t}}\mu _{i,j}^{\rm{e}}}}{{\left( {\chi _{i,j}^{\rm{l}} + \mu _{i,j}^{\rm{t}}} \right)\left( {\chi _{i,j}^{\rm{l}} + \mu _{i,j}^{\rm{e}}} \right)}}\left( {\frac{1}{{\omega _{i,j}^{\rm{t}}}} + \frac{1}{\gamma_{i,j}}} \right)\right.\\
			& + \frac{{\chi _{i,j}^{\rm{l}}\left( \chi_{i,j}^{\rm{l}}+\mu_{i,j}^{\rm{t}}+\mu_{i,j}^{\rm{e}} \right)}}{{\left( {\chi _{i,j}^{\rm{l}} + \mu _{i,j}^{\rm{t}}} \right)\left( {\chi _{i,j}^{\rm{l}} + \mu _{i,j}^{\rm{e}}} \right)}}\left( {\frac{{\mu _{i,j}^{\rm{e}}}}{{\gamma_{i,j}\left( \gamma_{i,j}+\mu_{i,j}^{\rm{e}} \right)}} + \frac{1}{{\Omega _{i,j}^{\rm{t,e} }}}} \right) -\left. \frac{{\mu _{i,j}^{\rm{e}}\left(\mu_{i,j}^{\rm{e}}-\beta_{i,j}\xi_{i,j}\right)}}{{\gamma_{i,j}\left[ { \left(\omega _{i,j}^{\rm{e}}+\mu _{i,j}^{\rm{e}}\right)\gamma_{i,j} + \mu_{i,j}^{\rm{e}}\left(\mu_{i,j}^{\rm{e}}-\beta_{i,j}\xi_{i,j}\right)} \right]}}\right].
		\end{aligned}
	\end{equation}
\end{figure*}

Based on the result in Theorem \ref{theorem1}, the closed-form expression of MAoI under the local computing scheme can be obtained, which is given in Corollary \ref{corollary1}.
\begin{mycorollary}\label{corollary1}
In the large-scale MEC networks, the closed-form expression of MAoI $\bar \Delta^{\rm{l}}_{i,j}$ under the local computing scheme with the FCFS policy can be expressed by
\begin{equation}\label{eqn-local}
	{\bar \Delta^{\rm{l}}_{i,j}}= \frac{\xi _{i,j}^2}{{{{\left( {\mu _{i,j}^{\rm{l}}} \right)}^2}\left( {\mu _{i,j}^{\rm{l}} - {\xi _{i,j}}} \right)}} + \frac{1}{{\mu _{i,j}^{\rm{l}}}} + \frac{1}{{{\xi _{i,j}}}}.
\end{equation}
\end{mycorollary}
\begin{myproof}
By substituting $\beta_{i,j}=0$ into \eqref{eqn-14}, the equation \eqref{eqn-local} can be obtained. 
\end{myproof}

Similarly, the closed-form expression of MAoI under the remote computing scheme can be obtained, which is given in Corollary \ref{thm-remote} as follows.
\begin{mycorollary}\label{thm-remote}
In the large-scale MEC networks, the closed-form result of MAoI $\bar \Delta^{\rm{e}}_{i,j}$ under the remote computing scheme with the FCFS policy can be expressed by
\begin{equation}\label{eqn-remote}
	\begin{aligned}
		{\bar \Delta^{\rm{e}}_{i,j}}
		&=\frac{{{\xi _{i,j}^2}\left[ {\left( {\mu _{i,j}^{\rm{t}} + \mu _{i,j}^{\rm{e}}} \right)\left( {\mu _{i,j}^{\rm{t}} + \mu _{i,j}^{\rm{e}} - {\xi _{i,j}}} \right) - \mu _{i,j}^{\rm{t}}\mu _{i,j}^{\rm{e}}} \right]}}{{\mu _{i,j}^{\rm{t}}{{\left( {\mu _{i,j}^{\rm{e}}} \right)}^2}\left( {\mu _{i,j}^{\rm{e}} - {\xi _{i,j}}} \right)\left( {\mu _{i,j}^{\rm{t}} + \mu _{i,j}^{\rm{e}} - {\xi _{i,j}}} \right)}}\\
		& + \frac{\xi _{i,j}^2}{{{{\left( {\mu _{i,j}^{\rm{t}}} \right)}^2}\left( {\mu _{i,j}^{\rm{t}} - {\xi _{i,j}}} \right)}} + \frac{1}{{\mu _{i,j}^{\rm{t}}}} + \frac{1}{{{\xi _{i,j}}}} + \frac{1}{{\mu _{i,j}^{\rm{e}}}}.
	\end{aligned}
\end{equation}
\end{mycorollary}
\begin{myproof} 
By substituting $\beta_{i,j}=1$ into \eqref{eqn-14}, the equation \eqref{eqn-remote} can be obtained.
\end{myproof}

For the partial offloading scheme, it can be observed that the MAoI performance is affected by the COR and TGR. Therefore, to analyse the effects of COR and TGR on the MAoI performance in the large-scale partial offloading MEC networks, we formulate an optimisation problem to minimise the MAoI given in Theorem \ref{theorem1} as follows:
\begin{equation}
\begin{split}
&\arg \min_{\{ {\beta _{i,j}}\},\{{\xi_{i,j}}\}} \bar \Delta _{i,j}^{{\rm{p}}}\\
&\begin{array}{l@{\quad}l@{}l@{\quad}r}
{\rm{s.t.}} &{\rm{C1}}: 0\leq{\beta _{i,j}}\leq1\\
		&{\rm{C2}}: (1-\beta_{i,j})\xi_{i,j}\leq\mu_{i,j}^{\rm{l}}\\
		&{\rm{C3}}: \beta_{i,j}\xi_{i,j}\leq\mu_{i,j}^{\rm{t}}\\
		&{\rm{C4}}: \beta_{i,j}\xi_{i,j}\leq\mu_{i,j}^{\rm{e}},\\
\end{array}
\end{split}
\end{equation}
where C1 gives the feasible range of COR. C2, C3, and C4 provide the feasible ranges of TGR at the local processor, the transmitter, and the edge server to ensure the stability of the Jackson network, respectively. Since the MAoI has been derived in closed-form and the constraints of COR and TGR are limited, the optimal COR and TGR for minimising the MAoI can be obtained by using the bisection method \cite{ref15}.

\section{Numerical results and analysis}\label{sec_four}
In this section, the theoretical MAoI expressions for the different computing offloading schemes are validated by the Monte Carlo simulations. Based on the derived expressions, the effects of COR, TGR, and computing capability of UEs in terms of the MAoI are analysed numerically in the large-scale partial offloading MEC networks.

\subsection{Validation}
In each simulation iteration, we deploy the locations of the BSs and UEs following the MCP described in Section \ref{sec_two}. Subsequently, according to the equation \eqref{eqn-1}, we calculate the SIR of the typical BS in each iteration. As a result, the STP can be obtained by calculating the cumulative distribution function (CDF) of SIR with 10,000 iterations. Equipped with the STP, the results of MAoI can be validated. 
Subsequently, we model 10,000 arrival tasks following the Poisson process in the Jackson network. The service time of each task is modelled as the exponential distribution with parameters related to the service rate in each queueing model of the Jackson network. For each arrival task, we record the arrival time and the departure time in each queueing model of the Jackson network. Thus, the simulation results of MAoI can be obtained by calculating the mathematical conditional expectation in \eqref{eqn-9}.  
Note that the computing capability of UEs depends on their type. Since UEs with different computing capabilities have different MAoI results, the evaluation of MAoI performance is complicated. To simplify the analysis, we assume that each UE has the same values for TGR, COR, and computing frequency. Meanwhile, it is also assumed that the number of UEs in each BS is the same and that the computing frequency of the BS are equally allocated to all UEs in the BS \cite{ref24}. The parameter values are listed in Table \ref{tab-2} unless otherwise specified.
\begin{table}[!b]
	\centering
	\caption{Simulation Values}
	\label{tab-2}
	\begin{tabular}{|c|c|c|c|}
		\hline
		& & & \\[-6pt] 
		Parameters & Values & Parameters & Values \\
		\hline
		& & & \\[-6pt] 
		$\tau$ & 0 dB & $\alpha$& 4 \\
		\hline
		& & &\\ [-6pt]
		$\epsilon$ & 0.5 & ${\rm{N}}_i$ & 20 \\
		\hline
		& & &\\ [-6pt]
		$B_{\rm{tot}}$ & 50 MHz & $L_{i,j}$ & 2 Mbits \\
		\hline
		& & &\\ [-6pt]
		$f_{i,j}$ & 1 GHz & $f^{\rm{B}}_{i}$ & 45 GHz \\
		\hline
		& & &\\ [-6pt]
		$C_{i,j}$ & 900 CPU cycles/bit& $\beta_{i,j}$ & 0.4 \\
		\hline
		& & &\\ [-6pt]
		$\xi_{i,j}$ & 0.2 & $\lambda_{\rm{B}}$ & $1\times 10^{-4}$ BSs/$m^2$ \\
		\hline
	\end{tabular}
\end{table}

\begin{figure}[!t]
\centering
\includegraphics[width=2.8in]{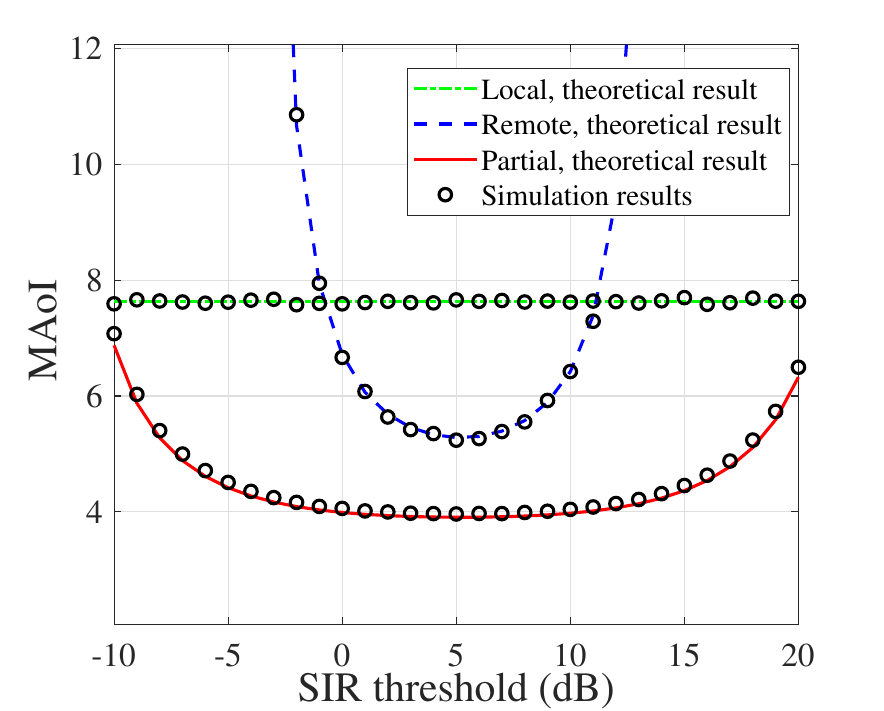}
\caption{The simulation and theoretical results of MAoI versus the SIR threshold for the local computing, remote computing, and partial offloading schemes.}\label{fig3}
\centering
\includegraphics[width=2.8in]{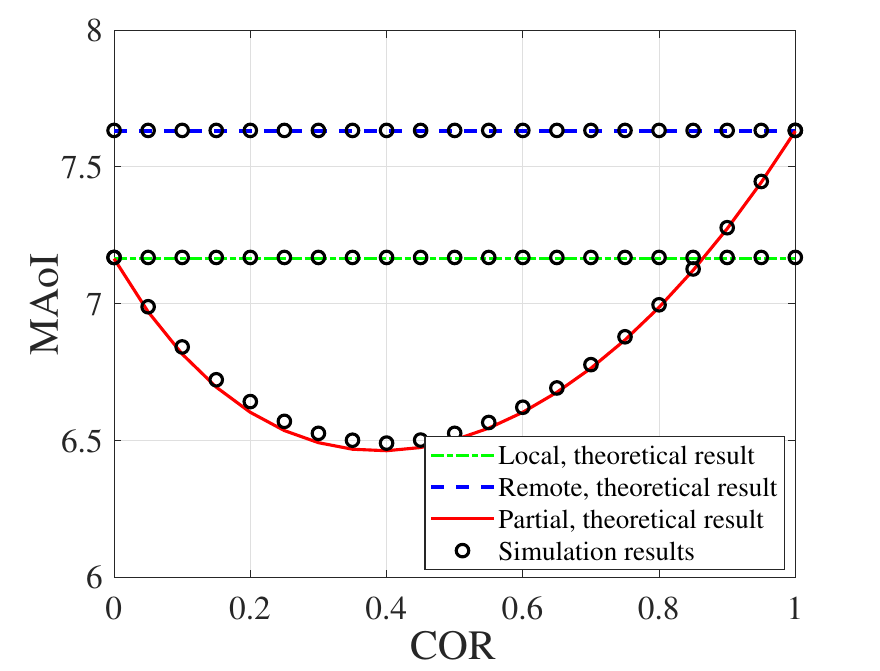}
\caption{The simulation and theoretical results of MAoI versus the COR for the local computing, remote computing, and partial offloading schemes.}\label{fig4}
\end{figure}
 
In Fig. \ref{fig3} and Fig. \ref{fig4}, the theoretical and simulation results of MAoI under the local computing, remote computing, and partial offloading schemes are plotted versus the SIR threshold and the COR, respectively. The results show that the theoretical MAoIs are matched with the simulation results, verifying the correctness of our derived closed-form expressions. Moreover, it is observed in Fig. \ref{fig4} that there exists an optimal COR to minimise the MAoI under the partial offloading scheme. Under our simulation environment, with the increasing of COR, a sharp decrease in terms of the MAoI occurs when the COR is in the range from 0 to 0.4, and then the MAoI gradually increases under the partial offloading scheme. It indicates that the COR has a significant impact on the MAoI performance under the partial offloading scheme and the optimal COR occurs near 0.4.  

\subsection{Analysis}
\begin{figure}[!t]
\centering
\includegraphics[width=2.9in]{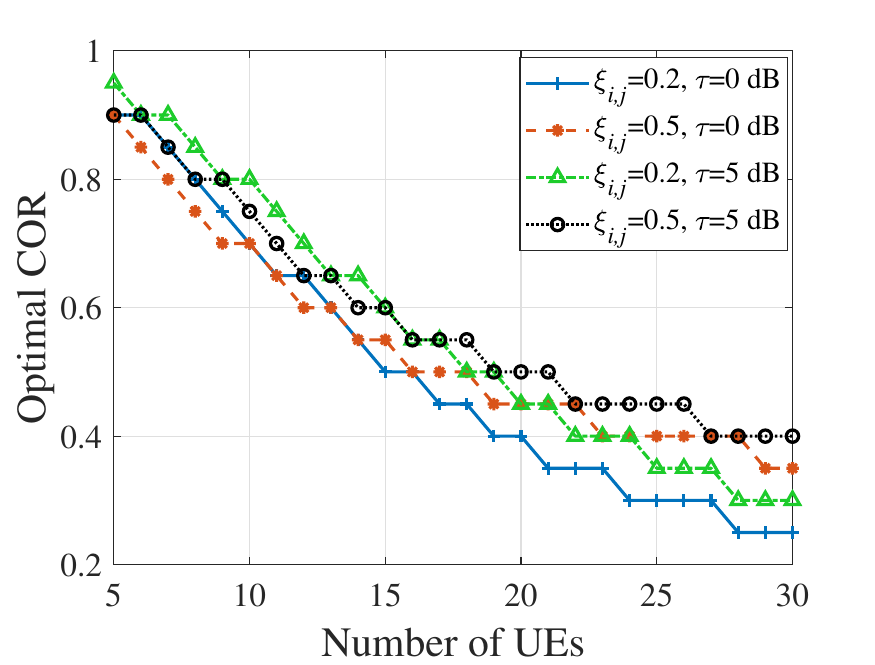}
\caption{The relationship between the optimal COR and the number of UEs with the TGR being 0.2 and 0.5 and the SIR threshold being 0 dB and 5 dB. }\label{fig5}
\end{figure}
In Fig. \ref{fig5}, we investigate the influence of the number of UEs on the optimal COR with the TGR being 0.2 and 0.5 and the SIR threshold being 0 and 5 dB. The results show that by increasing the number of UEs, the optimal COR under the partial offloading scheme declines. This is mainly because each UE is allocated fewer spectrum resources with the increasing of number of UEs. Therefore, offloading fewer computation-intensive tasks from the UE to its associated BS can reduce transmission delay, which decreases the MAoI. Furthermore, when the number of UEs is 30 and the SIR threshold is 0 dB, we compare the scenarios with $\xi_{i,j}=0.2$ and $\xi_{i,j}=0.5$. Since a higher TGR causes a higher computing burden on the UE, more computation-intensive tasks need to be offloaded to the BS from the UE to alleviate the computing burden on the UE, thus causing a significant improvement for the MAoI. Consequently, TGR has a significant impact on the optimal COR.

\begin{figure}[!t]
\centering
\includegraphics[width=2.8in]{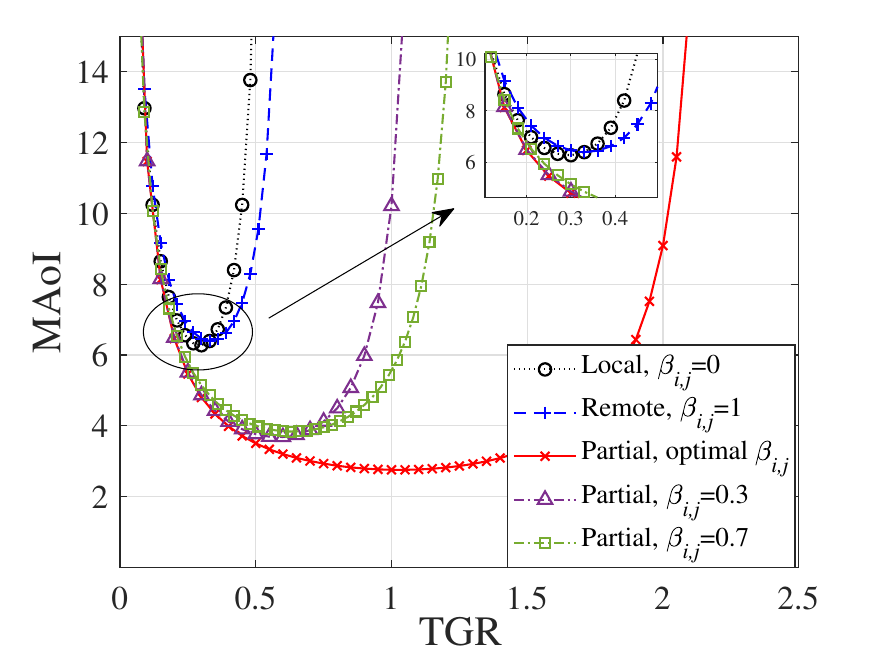}
\caption{The relationship between the MAoI and the TGR with the COR being 0, 0.3, 0.7, 1 and optimal. }\label{fig6}
\end{figure}
In Fig. \ref{fig6}, the MAoIs are illustrated versus the TGR with the COR $\beta_{i,j}$ being 0, 0.3, 0.7, 1, and optimal. It can be observed that with the increasing of TGR, a sharp decrease in terms of the MAoI occurs when the TGR is in the range from 0.3 to 1.2, and then the MAoI gradually increases based on our simulation environment. It indicates that the TGR has a significant impact on the MAoI performance. A smaller TGR results in a larger MAoI due to the server or processor being idle. Meanwhile, a larger TGR leads to a larger MAoI due to the extended waiting latency of computation-intensive tasks in queue. Furthermore, as compared with the local and remote computing schemes, i.e., $\beta_{i,j}=0$ and $\beta_{i,j}=1$, the MAoI with the optimal COR can be reduced by up to 58\% under the partial offloading scheme. As a result, the minimal MAoI can be obtained by jointly adjusting the optimal TRG and COR.

\begin{figure}[!t]
\centering
\includegraphics[width=2.8in]{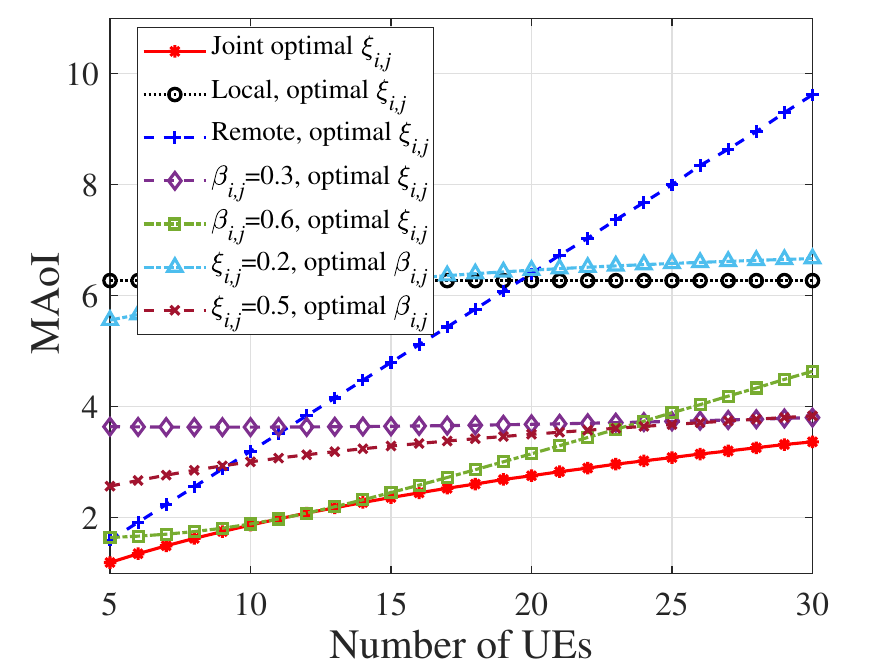}
\caption{The relationship between the MAoI and the number of UEs under the seven cases: joint optimal COR and TGR, optimal TGR with the COR being 0, 0.3, 0.6 and 1, optimal COR with the TGR being 0.2 and 0.5. }\label{fig7}
\end{figure}
In Fig. \ref{fig7}, we investigate the relationship between the MAoI and the number of UEs under different cases: 1) joint optimal OR and TGR; 2) optimal OR with the TGR being 0.2 and 0.5; 3) optimal TGR with the COR being 0, 0.3, 0.6 and 1. The results show that by adjusting the optimal TGR and COR, the partial offloading scheme outperforms the local and remote computing schemes in terms of the MAoI. The MAoI can be decreased by up to 51\% and 61\% under the partial offloading scheme with the optimal TGR and COR compared to the local and remote computing schemes when the number of UEs is 25, respectively. It indicates that the partial offloading scheme, which can fully utilize the computing resources in UEs and BSs, should be adopted to enhance the information freshness in the large-scale MEC networks.

\begin{figure}[!t]
\centering
\includegraphics[width=2.8in]{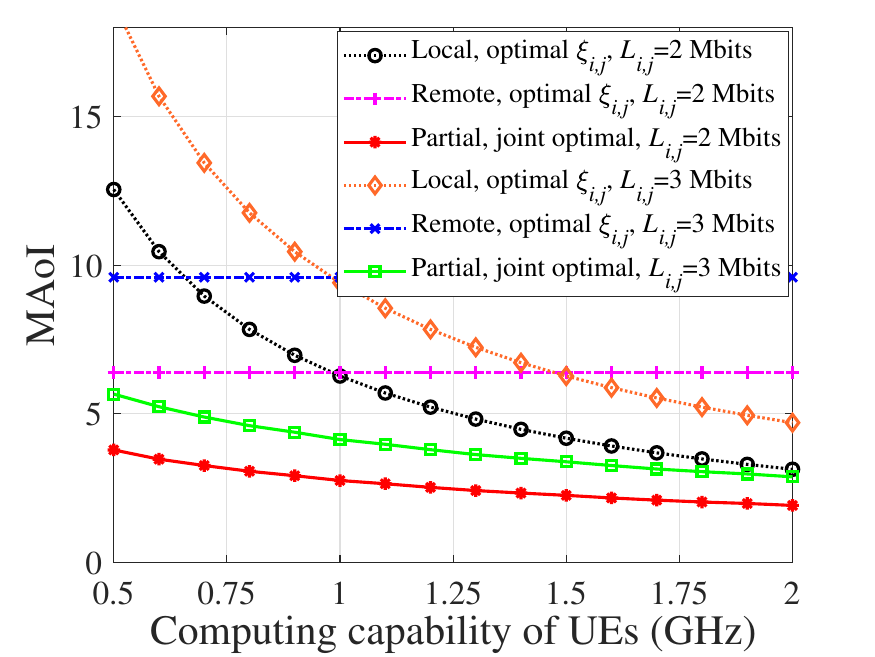}
\caption{The relationship between the MAoI and the computing capability of UEs with the average task size being 2 and 3 Mbits.}\label{fig8}
\end{figure}
Equipped with the optimal TGR and COR, in Fig. \ref{fig8}, the MAoIs are illustrated versus the computing capability of UE under the local computing, remote computing, and partial offloading schemes with the average size of tasks being 2 and 3 Mbits. It is found that by increasing the computing frequency of UE from 0.5 to 2 GHz, the MAoIs under the local computing scheme can be reduced by up to 9.5 when the average size of tasks is 2 Mbits, while the MAoI under the partial offloading scheme can be reduced by up to 3. This indicates that by employing the partial offloading scheme, the requirement on the computing capability of UEs can be reduced compared to the local computing scheme.

\section{Conclusion}\label{sec_five}
In this work, we have derived the closed-form expressions of MAoI in the large-scale partial offloading MEC networks. Based on these expressions, the effects of COR, TGR, and computing capability of UEs on MAoI performance are analysed numerically. It is observed that under our simulation environment, with optimal COR and TGR, the MAoI under the partial offloading scheme can be reduced by up to 51\% and 61\% compared to that under the local and remote computing schemes, respectively. Consequently, the partial offloading scheme should be applied in the large-scale MEC networks to enhance the information freshness. Furthermore, for systems requiring minimal MAoI, the pressure on UE computing capabilities can be alleviated by utilizing the partial offloading scheme. For future work, the joint optimisation of TGR and COR for different UE types should be considered to further improve MAoI performance in the large-scale partial offloading MSE networks.

\begin{appendices}
\section{Proof of Lemma \ref{lemma1}} \label{lp1}
According to \eqref{eqn-8}, the MAoI $\bar \Delta _{i,j}^{{\rm{p}}}$ under the partial offloading scheme can be transformed as 
\begin{equation}\label{eq22}
    \bar \Delta _{i,j}^{{\rm{p}}} = \mathop {\lim }\limits_{{\cal L} \to \infty } \left[\frac{{{Q_1} + {{\bar Q}_z}}}{{\cal L}} +\frac{{z - 2}}{{\cal L}}\frac{{\sum\nolimits_{n = 2}^{n=z - 1} {{Q_n}} }}{{z - 2}}\right].
\end{equation}
When $\mathcal{L}$ tends to infinity and the values of $Q_1$ and $\bar{Q}_n$ are finite, the term $\frac{{{Q_1} + {{\bar Q}_z}}}{{\cal L}}$ in \eqref{eq22} becomes zero. Besides, as shown in Fig. \ref{fig2}, the term $\frac{{z - 2}}{{\cal L}}$ in \eqref{eq22} can be expressed by $\frac{{z - 2}}{{\cal L}}=\frac{{z - 2}}{{{t_1} + \sum\nolimits_{n = 2}^z {{A_n}}  + T_z^{{\rm{max}}}}}$, where $A_n$ is the arrival interval between the $(n-1)$-th and $n$-th tasks at the sensor of UE, i.e., $A_n=t_n-t_{n-1}$. $T_z^{{\rm{max}}}$ is the system time of the $z$-th task under the partial offloading scheme. As the AoI decreases when the UE receives a complete computing result, the system time $T_n^{{\rm{max}}}$ for the $n$-th computation-intensive task is the maximum value between the system time $T_n^{\rm{l}}$ for local computing and the system time $T_n^{\rm{t,e}}$ for remote computing, i.e., $T_z^{{\rm{max}}}=\max\{T_z^{\rm{l}},T_z^{\rm{t,e}}\}$. Based on the law of large numbers \cite{ref32}, when the expectation of $A_n$ exists and the number of tasks $z$ is large enough, the arithmetic mean of $A_n$ can be approximately calculated by its mathematical expectation $\mathbb{E}[A_n]$. Moreover, $t_1$ and $T_z^{\rm{max}}$ are finite values with a probability of 1. Therefore, the term $\mathop {\lim }\limits_{{\cal L} \to \infty } \frac{{z - 2}}{{\cal L}}$ in \eqref{eq22} can be expressed by 
\begin{equation}\label{eq24}
    \mathop {\lim }\limits_{{\cal L} \to \infty } \frac{{z - 2}}{{\cal L}}=\frac{1}{\mathbb{E}{\left[ {{A_n}} \right]}}.
\end{equation}
For the term $\frac{{\sum\nolimits_{n = 2}^{n=z - 1} {{Q_n}} }}{{z - 2}}$ in \eqref{eq22}, the area $Q_n$ can be obtained by calculating the difference between two isosceles triangles, which is expressed by
\begin{equation}\label{eqn-12}
	\begin{aligned}
		{Q_n} &= \frac{1}{2}{\left( {{A_n} + T_n^{{\rm{max}}}} \right)^2} - \frac{1}{2}{\left( {T_n^{{\rm{max}}}} \right)^2} \\
		& = \frac{1}{2}{A_n ^2} + {A_n}T_n^{{\rm{max}}}.
	\end{aligned}
\end{equation}
As $z$ is approaching infinity, the term $\frac{{\sum\nolimits_{n = 2}^{n=z - 1} {{Q_n}} }}{{z - 2}}$ in \eqref{eq22} converges to its corresponding mathematical expectations similar to \eqref{eq24}. Consequently, the MAoI under the partial offloading scheme can be expressed by
	\begin{equation}\label{eqn-13}
		\bar \Delta _{i,j}^{{\rm{p}}} = \frac{1}{{\mathbb{E}\left[ {{A_n}} \right]}}\left( {\frac{1}{2}\mathbb{E}\left[ {A_n^2} \right] + \mathbb{E}\left[ {{A_n}T_n^{{\rm{max}}}} \right]} \right).
	\end{equation}
It can be observed that for the value of $T_n^{{\rm{max}}}$, there are two independent events, denoted by $E_n^{\rm{l}}$ and $E_n^{\rm{t,e}}$. For the event $E_n^{\rm{l}}$, the system time for local computing is longer than that for remote computing, i.e., $T_n^{\rm{l}}>T_n^{\rm{t,e}}$. Consequently, the system time $T_n^{{\rm{max}}}$ for the $n$-th computation-intensive task under the partial offloading scheme can be expressed by $T_n^{{\rm{max}}}=T_n^{\rm{l}}$. For the event $E_n^{\rm{t,e}}$, the system time for remote computing is longer than that for local computing, i.e., $T_n^{\rm{t,e}}>T_n^{\rm{l}}$. As a result, the system time $T_n^{{\rm{max}}}$ for the $n$-th computation-intensive task under the partial offloading scheme can be expressed by $T_n^{{\rm{max}}}=T_n^{\rm{t,e}}$. Moreover, recall that a portion $(1-\beta_{i,j})$ of the computation-intensive tasks is processed at the UE $\boldsymbol{{\rm{x}}}_{i,j}$, while the rest $\beta_{i,j}$ of them are offloaded to the BS $\boldsymbol{{\rm{y}}}_{i}$ for parallel computing. As a result, according to the law of total expectation \cite{ref22}, the MAoI can be calculated by
\begin{equation}\label{eqn14}
	\begin{aligned}
		\bar \Delta _{i,j}^{{\rm{p}}}& = \frac{1}{{\mathbb{E}\left[ {{A_n}} \right]}}\bigg( { \left( {1 - {\beta _{i,j}}} \right){\rm{P}}\left( {E_n^{\rm{l}}} \right)\mathbb{E}\left[ {A_n^{\rm{l}}T_n^{\rm{l}}\left| E_n^{\rm{l}} \right.} \right]}\\
		& +{{\beta _{i,j}}{\rm{P}}\left( {E_n^{\rm{t,e}}} \right)\mathbb{E}\left[ {A_n^{\rm{t}} {T_n^{\rm{t,e}}}\left| E_n^{\rm{t,e}} \right.} \right]} +\frac{1}{2}\mathbb{E}\left[ {{A_n}^2} \right] \!\!\bigg),
	\end{aligned}
\end{equation} 
which is given in Lemma \ref{lemma1}.

\section{Proof of Theorem \ref{theorem1}}\label{proof3}
It can be seen from \eqref{eqn-9} that $\mathbb{E}[A_n]$, $\mathbb{E}[{A_n}^2]$, ${\rm{P}}\left( {E_n^{\rm{l}}} \right)$, ${\rm{P}}\left( {E_n^{\rm{t,e}}} \right)$, $\mathbb{E}\left[ {A_n^{\rm{l}}T_n^{\rm{l}}| E_n^{\rm{l}}} \right]$, and $\mathbb{E}\left[ {A_n^{\rm{t}} {T_n^{\rm{t,e}}}| E_n^{\rm{t,e}}} \right]$ need to be derived to obtain the closed-form expression of MAoI under the partial offloading scheme. Since the arrival interval $A_n$ between the $(n-1)$-th and $n$-th computation-intensive tasks at the sensor of UE is an exponential random variable with parameter $\xi_{i,j}$, the expectations $\mathbb{E}[A_n]$ and $\mathbb{E}[A_n^2]$ can be calculated by $\mathbb{E}[A_n]=\frac{1}{\xi _{i,j}}$ and $\mathbb{E}[A_n^2]=\frac{2}{\xi _{i,j}^2}$, respectively.

Besides, the system time $T_n^{\rm{l}}$ of the $n$-th computation-intensive task for local computing follows an exponential distribution with parameter $(\mu^{\rm{l}}_{i,j}-(1-\beta_{i,j})\xi_{i,j})$ \cite{ref22}. As a result, the probability density function (PDF) $f_{T_n^{{\rm{l}}}}\left(t\right)$ of the system time $T_n^{\rm{l}}$ can be expressed by 
\begin{equation}\label{pdf_tnl}
    f_{T_n^{{\rm{l}}}}\left(t\right) \!=\! \left( {\mu _{i,j}^{\rm{l}} \!-\! \left( {1 \!-\! {\beta _{i,j}}} \right){\xi _{i,j}}} \right){e^{ - \left( {\mu _{i,j}^{\rm{l}} \!-\! \left( {1 \!-\! {\beta _{i,j}}} \right){\xi _{i,j}}} \right)t}}.
\end{equation}
The system time $T_n^{\rm{t,e}}$ of the $n$-th computation-intensive task for remote computing consists of the offloading delay $T_n^{\rm{t}}$ from the UE to its associated BS and the system time $T_n^{\rm{e}}$ for computing the task at the edge server, which can be expressed by
\begin{equation}\label{tt}
    T_n^{\rm{t,e}}=T_n^{\rm{t}}+T_n^{\rm{e}}.
\end{equation}
Therefore, the PDF ${f_{T_n^{\rm{t,e}}}}\left( t \right)$ of $T_n^{\rm{t,e}}$ can be calculated by
\begin{equation}
	\begin{aligned}\label{eqn-17}
		{f_{T_n^{\rm{t,e}}}}\left( t \right) &= {f_{T_n^{\rm{t}}}}\left( t \right) \otimes {f_{T_n^{\rm{e}}}}\left( t \right)\\
		&= \eta_{i,j}\left( {{e^{ - \left( {\mu _{i,j}^{\rm{t}} \!\!-\! {\beta _{i,j}}{\xi _{i,j}}} \right)t}}} \right.
		\left. { - {e^{ - \left( {\mu _{i,j}^{\rm{e}} \!\!-\! {\beta _{i,j}}{\xi _{i,j}}} \right)t}}} \right),
	\end{aligned}
\end{equation}
where $\otimes$ is the  convolution operation. ${f_{T_n^{\rm{t}}}}\left( t \right)$ and ${f_{T_n^{\rm{e}}}}\left( t \right)$ are the PDFs of $T_n^{\rm{t}}$ and $T_n^{\rm{e}}$, which follow exponential distributions with parameters $(\mu^{\rm{t}}_{i,j}-\beta_{i,j}\xi_{i,j})$ and $(\mu^{\rm{e}}_{i,j}-\beta_{i,j}\xi_{i,j})$, respectively, similar to \eqref{pdf_tnl}. Additionally, $\eta_{i,j} = \frac{{\left( {\mu _{i,j}^{\rm{t}} \!\!-\! {\beta _{i,j}}{\xi _{i,j}}} \right)\left( {\mu _{i,j}^{\rm{e}} \!\!-\! {\beta _{i,j}}{\xi _{i,j}}} \right)}}{{\mu _{i,j}^{\rm{e}} \!\!-\! \mu _{i,j}^{\rm{t}}}}$. As a result, the probability $ {\rm{P}}\left(E_n^{\rm{l}}\right)$ in \eqref{eqn-9} can be calculated by
\begin{equation}
    \begin{aligned}\label{eq_el}
        {\rm{P}}\left(E_n^{\rm{l}}\right) &= \int_0^\infty  {f_{T_n^{\rm{l}}}}\left( t \right){{\rm{P}}\left( {t > T_n^{\rm{t,e}}} \right)}{\rm{d}}t \\
        &= \int_0^\infty  {\int_0^t {f_{T_n^{\rm{l}}}}\left( t \right){{f_{T_n^{\rm{t,e}}}}\left( m \right)} } {\rm{d}}m{\rm{d}}t\\
        &=\frac{{\left( {\mu _{i,j}^{\rm{t}} \!\!-\! {\beta _{i,j}}{\xi _{i,j}}} \right)\left( {\mu _{i,j}^{\rm{e}} \!\!-\! {\beta _{i,j}}{\xi _{i,j}}} \right)}}{{\left( {\mu _{i,j}^{\rm{t}} \!\!+\! \mu _{i,j}^{\rm{l}} \!\!-\! {\xi _{i,j}}} \right)\left( {\mu _{i,j}^{\rm{e}} \!\!+\! \mu _{i,j}^{\rm{l}} \!\!-\! {\xi _{i,j}}} \right)}}.
    \end{aligned}
\end{equation}
Since $E_n^{\rm{t,e}}$ is a complementary event to the event $E_n^{\rm{l}}$, the probability ${\rm{P}}\left(E_n^{\rm{t,e}}\right)$ in \eqref{eqn-9} can be calculated by
\begin{equation}\label{PEM}
	\begin{aligned}
		&{\rm{P}}\left(E_n^{\rm{t,e}}\right)=1-{\rm{P}}\left(E_n^{\rm{l}}\right)\\
		&=\frac{\left(\mu_{i,j}^{\rm{l}}\!-\!\left(1-\beta_{i,j}\right)\xi_{i,j}\right)\left(\mu_{i,j}^{\rm{l}}\!+\!\mu_{i,j}^{\rm{t}}\!+\!\mu_{i,j}^{\rm{e}}\!-\!\left(1\!-\!\beta_{i,j}\right)\xi_{i,j}\right)}{{\left( {\mu _{i,j}^{\rm{l}} \!\!+\! \mu _{i,j}^{\rm{t}} \!\!-\! {\xi _{i,j}}} \right)\left( {\mu _{i,j}^{\rm{l}} \!\!+\! \mu _{i,j}^{\rm{e}} \!\!-\! {\xi _{i,j}}} \right)}}.
	\end{aligned}
\end{equation}

Additionally, it is challenging to directly calculate the expectation $\mathbb{E}\left[ {A_n^{\rm{l}}T_n^{\rm{l}}| E_n^{\rm{l}}} \right]$. Specifically, the system time $T_n^{\rm{l}}$ of the $n$-th computation-intensive task for local computing is related to its arrival interval $A_n^{\rm{l}}$ at the local processor of UE. A larger arrival interval can shorten the system time due to a shorter waiting time. Conversely, a smaller arrival interval leads to a longer system time due to tasks accumulating in the queue, which increases the waiting time. In other words, the arrival interval is related to the waiting time and is independent of the service time. Therefore, the system time $T_n^{\rm{l}}$ of the $n$-th computation-intensive task for local computing can be regarded as the sum of its waiting time $W_n^{\rm{l}}$ and its service time $S_n^{\rm{l}}$, which can be expressed by 
\begin{equation}\label{eq31}
    T_n^{\rm{l}} = S_n^{\rm{l}}+W_n^{\rm{l}}.
\end{equation}
If the $n$-th computation-intensive task arrives at the processing queue of the local processor while the $\left(n-1\right)$-th task is still waiting or served in the queue, then $W_n^{\rm{l}}=T_{n-1}^{\rm{l}}-A_n^{\rm{l}}$. Otherwise, $W_n^{\rm{l}}=0$. Consequently, the waiting time $W_n^{\rm{l}}$ of the $n$-th computation-intensive task for local computing can be expressed by
\begin{equation}\label{eqn-wnl}
	W_n^{\rm{l}} = \left\{
	\begin{aligned}
		&T_{n-1}^{\rm{l}}-A_n^{\rm{l}}, &T_{n-1}^{\rm{l}}>A_n^{\rm{l}}\\
		&0, &T_{n-1}^{\rm{l}}<A_n^{\rm{l}}.
	\end{aligned}
	\right.
\end{equation} 
Therefore, according to the law of total expectation \cite{ref33}, the conditional expectation $\mathbb{E}\left[A_n^{\rm{l}}T_n^{\rm{l}}|E_n^{\rm{l}}\right]$ can be transformed as 
\begin{equation}\label{eqn-el}
\begin{aligned}
	&\mathbb{E}\left[A_n^{\rm{l}}T_n^{\rm{l}}|E_n^{\rm{l}}\right]=\mathbb{E}\left[A_n^{\rm{l}}\left(W_n^{\rm{l}}+S_n^{\rm{l}}\right)|E_n^{\rm{l}}\right]\\
    &= {\rm{P}}\left( {T_{n - 1}^{\rm{l}} \!>\! A_n^{\rm{l}}} \right)\left\{\mathbb{E} \left[ {A_n^{\rm{l}}{T_{n - 1}^{\rm{l}}} |E_n^{\rm{l}},T_{n - 1}^{\rm{l}} \!>\! A_n^{\rm{l}}} \right]\right.\\
    &\left.-\mathbb{E} \left[ ({A_n^{\rm{l}})^2 |E_n^{\rm{l}},T_{n - 1}^{\rm{l}} \!>\! A_n^{\rm{l}}} \right]\right\}\\
    &+\mathbb{E}\left[ {S_n^{\rm{l}}|E_n^{\rm{l}}}\right]\left\{ {\rm{P}}\left( {T_{n - 1}^{\rm{l}} \!>\! A_n^{\rm{l}}} \right)\mathbb{E}\left[ {A_n^{\rm{l}}|E_n^{\rm{l}},T_{n - 1}^{\rm{l}} \!>\! A_n^{\rm{l}}} \right]\right.\\
    &\left. + {\rm{P}}\left( {T_{n - 1}^{\rm{l}} \!<\! A_n^{\rm{l}}} \right)\mathbb{E}\left[ {A_n^{\rm{l}}|E_n^{\rm{l}},T_{n - 1}^{\rm{l}} \!<\! A_n^{\rm{l}}} \right] \right\} \triangleq \Xi_{i,j}^{1}.
\end{aligned}
\end{equation}
As $T_{n-1}^{\rm{l}}$ and $A_n^{\rm{l}}$ are independent, the probability ${\rm{P}}(T_{n-1}^{\rm{l}}>A_n^{\rm{l}})$ in \eqref{eqn-el} can be derived by
\begin{equation}\label{eqn-w}
	\begin{aligned}
		&{\rm{P}}\left( {T_{n - 1}^{\rm{l}} > A_n^{\rm{l}}} \right) = \int_0^\infty  {{f_{T_{n - 1}^{\rm{l}}}}\left( t \right){\rm{P}}\left( {t > A_n^{\rm{l}}} \right)} {\rm{d}}t\\
		&= \int_0^\infty  {\int_0^t {{f_{T_{n - 1}^{\rm{l}}}}\left( t \right){f_{A_n^{\rm{l}}}}\left( x \right)} } {\rm{d}}x{\rm{d}}t = \frac{{\left( {1 - {\beta _{i,j}}} \right){\xi _{i,j}}}}{{\mu _{i,j}^{\rm{l}}}},
	\end{aligned}
\end{equation}
where the PDF ${f_{T_{n - 1}^{\rm{l}}}}\left( t \right)$ is given in \eqref{pdf_tnl}, and the PDF $f_{A_n^{\rm{l}}}(x)$ of arrival interval $A_n^{\rm{l}}$ at the local processor of UE is expressed by $f_{A_n^{\rm{l}}}\left( x \right)=(1-\beta_{i,j})\xi_{i,j}e^{-(1-\beta_{i,j})\xi_{i,j}x}$. As the event $({T_{n - 1}^{\rm{l}} < A_n^{\rm{l}}})$ is a complementary event to the event $({T_{n - 1}^{\rm{l}} < A_n^{\rm{l}}})$, the probability ${\rm{P}}(T_{n-1}^{\rm{l}}>A_n^{\rm{l}})$ in \eqref{eqn-el} can be calculated by
\begin{equation}\label{eqn_36}
    {\rm{P}}\left( {T_{n - 1}^{\rm{l}} \!<\! A_n^{\rm{l}}} \right) \!=\! 1\!-\!{\rm{P}}\left( {T_{n - 1}^{\rm{l}} \!>\! A_n^{\rm{l}}} \right)\!=\!1\!-\!\frac{{\left( {1 - {\beta _{i,j}}} \right){\xi _{i,j}}}}{{\mu _{i,j}^{\rm{l}}}}.
\end{equation}

Besides, based on the equations \eqref{eq31} and \eqref{eqn-wnl}, the term $\mathbb{E} \left[ {A_n^{\rm{l}}{T_{n - 1}^{\rm{l}}} |E_n^{\rm{l}},T_{n - 1}^{\rm{l}} \!>\! A_n^{\rm{l}}} \right]$ in \eqref{eqn-el} can be transformed as
\begin{equation}\label{eq35}
    \begin{aligned}
	&\mathbb{E}[A_n^{\rm{l}}T_{n - 1}^{\rm{l}}|E_n^{\rm{l}},T_{n - 1}^{\rm{l}} > A_n^{\rm{l}}] \\
	&=\mathbb{E}[A_n^{\rm{l}}T_{n - 1}^{\rm{l}}|T_{n - 1}^{\rm{l}} - A_n^{\rm{l}} > T_n^{\rm{t,e}} - S_n^{\rm{l}},T_{n - 1}^{\rm{l}} > A_n^{\rm{l}}]\\
	&= {\rm{P}}\left( {{Y_n} > 0} \right)\mathbb{E}[A_n^{\rm{l}}T_{n - 1}^{\rm{l}}|T_{n - 1}^{\rm{l}} > A_n^{\rm{l}} + {Y_n},{Y_n} > 0] \\
	&+ {\rm{P}}\left( {{Y_n} < 0} \right)\mathbb{E}[A_n^{\rm{l}}T_{n - 1}^{\rm{l}}|T_{n - 1}^{\rm{l}} > A_n^{\rm{l}},{Y_n} < 0],
    \end{aligned}
\end{equation}
where $Y_n=T_n^{\rm{t,e}}-S_n^{\rm{l}}$. The PDF of $Y_n$ can be calculated by 
\begin{equation}\label{eqn-yn}
	\begin{aligned}
		&f_{Y_n}\left(y\right) = {f_{T_n^{\rm{t,e}}}}\left( y \right) \otimes {f_{S_n^{\rm{l}}}}\left( y \right)\\
		&=\mu_{i,j}^{\rm{l}}\eta_{i,j}\left\{
		\begin{aligned}
			&\frac{1}{{\Omega _{i,j}^{\rm{l,t} }}}{e^{ - \chi _{i,j}^{\rm{t}}y}} - \frac{1}{{\Omega _{i,j}^{\rm{l,e} }}}{e^{ - \chi _{i,j}^{\rm{e}}y}} &,y>0\\
			&\frac{{\mu _{i,j}^{\rm{e}} - \mu _{i,j}^{\rm{t}}}}{{\Omega _{i,j}^{{\rm{l,t}} }\Omega _{i,j}^{{\rm{l,e}} }}}{e^{\mu _{i,j}^{\rm{l}}y}} &,y<0,
		\end{aligned}
		\right.
	\end{aligned}
\end{equation}
where the PDF $f_{T_n^{\rm{t,e}}}(y)$ is given in \eqref{eqn-17}, and the PDF of $S_n^{\rm{l}}$ can be expressed by $f_{S_n^{\rm{l}}}(y) = \mu_{i,j}^{\rm{l}}e^{\mu_{i,j}^{\rm{l}}y}$. Additionally, $\chi _{i,j}^{\rm{e}}=\mu_{i,j}^{\rm{e}}\!-\!\beta_{i,j}\xi_{i,j}$. As a result, the probability ${\rm{P}}\left( {{Y_n} > 0} \right)$ in \eqref{eq35} can be expressed by  
\begin{equation}\label{eq37}
\begin{aligned}
&{\rm{P}}\left(Y_n>0\right)=\int_0^\infty  {{f_{{Y_n}}}\left( y \right)} {\rm{d}}y\\
&=\frac{{\mu _{i,j}^{\rm{l}}(\mu _{i,j}^{\rm{l}} + \mu _{i,j}^{\rm{t}} + \mu _{i,j}^{\rm{e}} - 2{\beta _{i,j}}{\xi _{i,j}})}}{{\Omega _{i,j}^{\rm{l,t} }\Omega _{i,j}^{\rm{l,e} }}}.
\end{aligned}
\end{equation}
As the event $\left(Y_n<0\right)$ is complementary to $\left(Y_n>0\right)$, the probability ${\rm{P}}\left(Y_n<0\right)$ in \eqref{eq35} can be expressed by
\begin{equation}\label{eq38}
    \begin{aligned}
        &{\rm{P}}\left(Y_n<0\right)=1-{\rm{P}}\left(Y_n>0\right)\\
        &=1-\frac{{\mu _{i,j}^{\rm{l}}(\mu _{i,j}^{\rm{l}} + \mu _{i,j}^{\rm{t}} + \mu _{i,j}^{\rm{e}} - 2{\beta _{i,j}}{\xi _{i,j}})}}{{\Omega _{i,j}^{\rm{l,t} }\Omega _{i,j}^{\rm{l,e} }}}.
    \end{aligned}
\end{equation}

Additionally, since $S_n^{\rm{l}}$ is independent from $A_n^{\rm{l}}$ and $T_{n-1}^{\rm{l}}$, $Y_n$, $A_n^{\rm{l}}$, and $T_{n-1}^{\rm{l}}$ are independent of each other. Therefore, the conditional expectation $\mathbb{E}[A_n^{\rm{l}}T_{n - 1}^{\rm{l}}|T_{n - 1}^{\rm{l}} > A_n^{\rm{l}} + {Y_n},{Y_n} > 0]$ in \eqref{eq35} can be calculated by
\begin{equation}\label{eq39}
\begin{aligned}
&\mathbb{E}[A_n^{\rm{l}}T_{n - 1}^{\rm{l}}|T_{n - 1}^{\rm{l}} > A_n^{\rm{l}} + {Y_n},{Y_n} > 0]\\
&=\frac{\mathbb{E}{[A_n^{\rm{l}}T_{n - 1}^{\rm{l}},T_{n - 1}^{\rm{l}} > A_n^{\rm{l}} + {Y_n},{Y_n} > 0]}}{{{\rm{P}}\left( {T_{n - 1}^{\rm{l}} > A_n^{\rm{l}} + {Y_n},{Y_n} > 0} \right)}}  \\
&=\frac{{\int_0^\infty  {\int_0^\infty  {\int_{y + x}^\infty  {xt{f_{T_{n - 1}^{\rm{l}}}}(t)} {f_{A_n^{\rm{l}}}}(x){f_{{Y_n}}}(y){\rm{d}}t{\rm{d}}x{\rm{d}}y} } }}{{\int_0^\infty  {\int_0^\infty  {\int_{y + x}^\infty  {{f_{T_{n - 1}^{\rm{l}}}}(t)} {f_{A_n^{\rm{l}}}}(x){f_{{Y_n}}}(y){\rm{d}}t{\rm{d}}x{\rm{d}}y} } }}\\
& = \frac{{\mu _{i,j}^{\rm{l}} + 2\chi _{i,j}^{\rm{l}}}}{{(\mu {{_{i,j}^{\rm{l}}})^2}\chi _{i,j}^{\rm{l}}}} + \frac{{\frac{{\Omega _{i,j}^{\rm{l,e} }\omega _{i,j}^{\rm{e}}}}{{\omega _{i,j}^{\rm{t}}}} - \frac{{\Omega _{i,j}^{\rm{l,t} }\omega _{i,j}^{\rm{t}}}}{{\omega _{i,j}^{\rm{e}}}}}}{{\Omega _{i,j}^{\rm{l,e}}\omega _{i,j}^{\rm{e}} - \Omega _{i,j}^{\rm{l,t} }\omega _{i,j}^{\rm{t}}}}.
\end{aligned}
\end{equation}
Similarly, the conditional expectation $\mathbb{E}[A_n^{\rm{l}}T_{n - 1}^{\rm{l}}|T_{n - 1}^{\rm{l}} > A_n^{\rm{l}},{Y_n} < 0]$ in \eqref{eq35} can be expressed by
\begin{equation}\label{eq40}
\begin{aligned}
&\mathbb{E}[A_n^{\rm{l}}T_{n - 1}^{\rm{l}}|T_{n - 1}^{\rm{l}}\! >\! A_n^{\rm{l}},{Y_n} \!<\! 0]\!=\!\mathbb{E}[A_n^{\rm{l}}T_{n - 1}^{\rm{l}}|T_{n - 1}^{\rm{l}} \!>\! A_n^{\rm{l}}]\\
&=\frac{{\int_0^\infty  {\int_0^{\rm{t}} {xt{f_{T_{n - 1}^{\rm{l}}}}(t)} {f_{A_n^{\rm{l}}}}(x){\rm{d}}x{\rm{d}}t} }}{{{\rm{P}}\left( {T_{n - 1}^{\rm{l}} > A_n^{\rm{l}}} \right)}}=\frac{{\mu _{i,j}^{\rm{l}} + 2\chi _{i,j}^{\rm{l}}}}{{(\mu {{_{i,j}^{\rm{l}}})^2}\chi _{i,j}^{\rm{l}}}}.
\end{aligned}
\end{equation}

By incorporating $\eqref{eq37}$-$\eqref{eq40}$ into $\eqref{eq35}$, we can obtain the result of $\mathbb{E} \left[ {A_n^{\rm{l}}{T_{n - 1}^{\rm{l}}}|E_n^{\rm{l}},T_{n - 1}^{\rm{l}} \!>\! A_n^{\rm{l}}} \right]$ in \eqref{eqn-el}. Similar to the derivation steps in \eqref{eq35}-\eqref{eq40}, the term $\mathbb{E}[ {{(A_n^{\rm{l}}})^2 |E_n^{\rm{l}},T_{n - 1}^{\rm{l}} \!>\! A_n^{\rm{l}}}]$ in \eqref{eqn-el} can be obtained.

By substituting \eqref{eq31}, the event $E_n^{\rm{l}}$ can be transformed as $S_n^{\rm{l}}>X_n$, where $X_n=T_n^{\rm{t,e}}-W_n^{\rm{l}}$. The PDF of $X_n$ can be calculated by
\begin{equation}
	\begin{aligned}
		&{f_{{X_n}}}\left( x \right) = {f_{T_n^{\rm{t,e}}}}\left( x \right) \otimes {f_{W_n^{\rm{l}}}}\left( x \right) \\
		&= \eta _{i,j}\rho _{i,j}^{\rm{l}}\left\{
		\begin{aligned}
			&{\frac{{\Omega _{i,j}^{{\rm{l,t}} }}}{{\omega _{i,j}^{\rm{t}}}}{e^{ - \chi_{i,j}^{\rm{t}}x}} - \frac{{\Omega _{i,j}^{{\rm{l,e}}}}}{{\omega _{i,j}^{{\rm{e}}}}}{e^{ - \chi_{i,j}^{\rm{e}}x}}} &, x>0\\
			&\frac{{\chi _{i,j}^{{\rm{l}}}\left( {\mu _{i,j}^{\rm{e}} - \mu _{i,j}^{\rm{t}}} \right)}}{{\omega _{i,j}^{\rm{t}}\omega _{i,j}^{{\rm{e}}}}}{e^{\chi _{i,j}^{{\rm{l}}}x}}&, x<0,
		\end{aligned}
		\right.
	\end{aligned}
\end{equation}
where the PDF ${f_{T_n^{\rm{t,e}}}}\left( x \right)$ is given in \eqref{eqn-17}. The PDF of $W_n^{\rm{l}}$ is expressed by ${f_{W_n^{\rm{l}}}}\left( x \right) = \left( {1 - \rho _{i,j}^{\rm{l}}} \right)\delta \left( x \right) + \chi _{i,j}^{\rm{l}}\rho _{i,j}^{\rm{l}}{e^{ - \chi _{i,j}^{\rm{l}}x}}$, where $\rho _{i,j}^{\rm{l}}=\frac{(1-\beta_{i,j})\xi_{i,j}}{\mu_{i,j}^{\rm{l}}}$ and $\delta \left( x \right)$ is an impulse function \cite{ref29}. Therefore, the term $\mathbb{E}\left[ {S_n^{\rm{l}}|E_n^{\rm{l}}} \right]$ in \eqref{eqn-el} is calculated by
\begin{equation}\label{eqn-es}
	\begin{aligned}
		&\mathbb{E}\left[ {S_n^{\rm{l}}|E_n^{\rm{l}}} \right] = \mathbb{E}\left[ {S_n^{\rm{l}}|S_n^{\rm{l}} > {X_n}} \right] = \int_0^\infty  s{f_{S_n^{\rm{l}}|S_n^{\rm{l}} > {X_n}}}\left( s \right){\rm{d}}s\\
		&=\int_0^\infty  s\frac{{{\rm{d}}{\rm{P}}\left( {S_n^{\rm{l}} < s,S_n^{\rm{l}} > {X_n}} \right)}}{{{\rm{P}}\left( {S_n^{\rm{l}} > {X_n}} \right){\rm{d}}s}}{\rm{d}}s\\
  &=\int_0^\infty\left\{ \frac{s}{{{{\rm{P}}\left( {E_n^{\rm{l}}} \right){\rm{d}}s}}}{\rm{d}}\left[ \int_{ - \infty }^0 {\int_0^s {{f_{S_n^{\rm{l}}}}\left( m \right){f_{{X_n}}}\left( x \right){\rm{d}}m{\rm{d}}x} } \right.\right.\\
  &\left.\left. + \int_0^\infty  {\int_x^s {{f_{S_n^{\rm{l}}}}\left( m \right){f_{{X_n}}}\left( x \right){\rm{d}}m{\rm{d}}x} }  \right]\right\}{\rm{d}}s\\
  & =\frac{1}{{\mu _{i,j}^{\rm{l}}}} + \frac{{\chi _{i,j}^{\rm{l}}\left( {\Omega _{i,j}^{{\rm{l,t}} } + \omega _{i,j}^{\rm{l,e}}} \right)}}{{\mu _{i,j}^{\rm{l}}\Omega _{i,j}^{\rm{l,t} }\Omega _{i,j}^{\rm{l,e}}}}.
	\end{aligned}
\end{equation}

Besides, the expectation $\mathbb{E}\left[ {A_n^{\rm{l}}|E_n^{\rm{l}},T_{n - 1}^{\rm{l}} \!>\! A_n^{\rm{l}}} \right]$ in \eqref{eqn-el} can be derived by
\begin{equation}\label{eqn-ea1}
\begin{aligned}
&\mathbb{E}\left[ {A_n^{\rm{l}}|E_n^{\rm{l}},T_{n - 1}^{\rm{l}} > A_n^{\rm{l}}} \right] = \int_0^\infty  {a{f_{A_n^{\rm{l}}|E_n^{\rm{l}},T_{n - 1}^{\rm{l}} > A_n^{\rm{l}}}}\left( a \right)} {\rm{d}}a\\
&= \!\int_0^\infty  {a\frac{{{\rm{d}}{\rm{P}}\left( {A_n^{\rm{l}} < a,E_n^{\rm{l}},T_{n - 1}^{\rm{l}} > A_n^{\rm{l}}} \right)}}{{{\rm{P}}\left( {E_n^{\rm{l}},T_{n - 1}^{\rm{l}} > A_n^{\rm{l}}} \right){\rm{d}}a}}} {\rm{d}}a\\
& = \!\int_0^\infty  \!{a\frac{{{\rm{d}}{\rm{P}}\left( {A_n^{\rm{l}} \!<\! a,A_n^{\rm{l}} \!<\! T_{n - 1}^{\rm{l}} \!-\! T_n^{{\rm{t,e}}} \!+\! S_n^{\rm{l}},T_{n - 1}^{\rm{l}} \!>\! A_n^{\rm{l}}} \right)}}{{{\rm{P}}\left( {A_n^{\rm{l}} < T_{n - 1}^{\rm{l}} - {Y_n},T_{n - 1}^{\rm{l}} > A_n^{\rm{l}}} \right){\rm{d}}a}}} {\rm{d}}a,
\end{aligned}
\end{equation}
where the probability ${\rm{P}}\left( {A_n^{\rm{l}} < T_{n - 1}^{\rm{l}} - {Y_n},T_{n - 1}^{\rm{l}} > A_n^{\rm{l}}} \right)$ in \eqref{eqn-ea1} can be expressed by
\begin{small}
\begin{equation}\label{eqn-f1_1}
	\begin{aligned}
		&{\rm{P}}\left( {A_n^{\rm{l}} < T_{n - 1}^{\rm{l}} - {Y_n},T_{n - 1}^{\rm{l}} > A_n^{\rm{l}}} \right)\\
		&=\int_0^\infty\!\!\!\int_y^\infty  {{f_{W_n^{\rm{l}}}}} \left( x \right){f_{{Y_n}}}(y){\rm{d}}w{\rm{d}}y + \!\!\int_{\!-\! \infty }^0\!\int_0^\infty  {{f_{W_n^{\rm{l}}}}} \left( x \right){f_{{Y_n}}}(y){\rm{d}}x{\rm{d}}y\\
		&=\frac{{{\eta _{i,j}}\left( {1 - {\beta _{i,j}}} \right){\xi _{i,j}}}}{{\Omega _{i,j}^{{\rm{l,t}} }\Omega _{i,j}^{{\rm{l,e}} }}}\left( {\frac{{\Omega _{i,j}^{{\rm{l,e}} }}}{{\omega _{i,j}^{{\rm{t}}}}} - \frac{{\Omega _{i,j}^{{\rm{l,t}} }}}{{\omega _{i,j}^{{\rm{e}}}}} + \frac{{\mu _{i,j}^{\rm{e}} - \mu _{i,j}^{\rm{t}}}}{{\mu _{i,j}^{\rm{l}}}}} \right).
	\end{aligned}
\end{equation}
\end{small}%

The probability ${\rm{P}}( {A_n^{\rm{l}} \!\!<\! a,\!A_n^{\rm{l}} \!\!<\! T_{n - 1}^{\rm{l}} \!\!-\! T_n^{{\rm{t,e}}} \!\!+\! S_n^{\rm{l}},\!T_{n - 1}^{\rm{l}} \!\!>\! A_n^{\rm{l}}})$ in \eqref{eqn-ea1} can be expressed by
\begin{small}
\begin{equation}\label{eqn-f1_2}
	\begin{aligned}
		&{\rm{P}}\left( {A_n^{\rm{l}} \!<\! a,A_n^{\rm{l}} \!<\! T_{n - 1}^{\rm{l}} \!-\! T_n^{{\rm{t,e}}} \!+\! S_n^{\rm{l}},T_{n - 1}^{\rm{l}} \!>\! A_n^{\rm{l}}} \right)\\
		&={\rm{P}}\left( {A_n^{\rm{l}} < a,A_n^{\rm{l}} < T_{n - 1}^{\rm{l}} - {Y_n},T_{n - 1}^{\rm{l}} > A_n^{\rm{l}}} \right)\\
		&=\int_{ - \infty }^0  \int_0^a {\int_x^\infty  {{f_{T_{n - 1}^{\rm{l}}}}\left( t \right)} {f_{A_n^{\rm{l}}}}\left( x \right)} {f_{Y_n}}\left( y \right){\rm{d}}t{\rm{d}}x{\rm{d}}y\\
		&+\int_0^\infty  {\int_0^a {\int_{x + y}^\infty  {{f_{T_{n - 1}^{\rm{l}}}}\left( t \right)} {f_{A_n^{\rm{l}}}}\left( x \right)} {f_{Y_n}}\left( y \right){\rm{d}}t{\rm{d}}x{\rm{d}}y} \\
		&\!=\!\frac{{{\eta _{i,j}}\!\left(\! {1\! - \!{\beta _{i,j}}} \right){\xi _{i,j}}\!\left(\! {1 \!- \!{e^{ - \mu _{i,j}^{\rm{l}}a}}} \right)}}{{\Omega _{i,j}^{{\rm{l,t}}}\Omega _{i,j}^{{\rm{l,e}}}}} \left( {\frac{{\Omega _{i,j}^{{\rm{l,e}}}}}{{\omega _{i,j}^{\rm{t}}}}\! - \!\frac{{\Omega _{i,j}^{{\rm{l,t}}}}}{{\omega _{i,j}^{\rm{e}}}} } \!+\! \frac{{\mu _{i,j}^{\rm{e}} \!-\!\mu _{i,j}^{\rm{t}}}}{{\mu _{i,j}^{\rm{l}}}}\!\right).
	\end{aligned}
\end{equation}
\end{small}%
By incorporating $\eqref{eqn-f1_1}$ and $\eqref{eqn-f1_2}$ into $\eqref{eqn-ea1}$, the expectation $\mathbb{E}\left[ {A_n^{\rm{l}}|E_n^{\rm{l}},T_{n - 1}^{\rm{l}} > A_n^{\rm{l}}} \right]$ in \eqref{eqn-el} can be expressed by
\begin{equation}\label{eqn-a|e1_1}
	\begin{aligned}
		\mathbb{E}\left[ {A_n^{\rm{l}}|E_n^{\rm{l}},T_{n - 1}^{\rm{l}} > A_n^{\rm{l}}} \right]=\frac{1}{\mu_{i,j}^{\rm{l}}}.
	\end{aligned}
\end{equation}

According to equations \eqref{eq31} and \eqref{eqn-wnl}, when the event $(T_{n-1}^{\rm{l}}<A_n^{\rm{l}})$ occurs, the event $E_n^{\rm{l}}$ can be transformed as $(S_n^{\rm{l}}>T_{n}^{\rm{t,e}})$, which is independent of $A_n^{\rm{l}}$. Therefore, the conditional expectation $\mathbb{E}\left[ {A_n^{\rm{l}}|E_n^{\rm{l}},T_{n - 1}^{\rm{l}} \!<\! A_n^{\rm{l}}} \right]$ in \eqref{eqn-el} can be calculated by
\begin{equation}\label{eqn-a|e1_2}
	\begin{aligned}
		&\mathbb{E}\left[ {A_n^{\rm{l}}|E_n^{\rm{l}},T_{n - 1}^{\rm{l}} \!<\! A_n^{\rm{l}}} \right] = \mathbb{E}\left[ {A_n^{\rm{l}}|T_{n - 1}^{\rm{l}} \!<\! A_n^{\rm{l}}} \right]\\
		&=\frac{{\int_0^\infty  {\int_0^x {{f_{T_{n - 1}^{\rm{l}}}}(t){f_{A_n^{\rm{l}}}}(x){\rm{d}}t{\rm{d}}x} } }}{{{\rm{P}}\left( {T_{n - 1}^{\rm{l}} < A_n^{\rm{l}}} \right)}}=\frac{{\mu _{i,j}^{\rm{l}} + (1 - {\beta _{i,j}}){\xi _{i,j}}}}{{\mu _{i,j}^{\rm{l}}(1 - {\beta _{i,j}}){\xi _{i,j}}}},
	\end{aligned}
\end{equation}

By incorporating \eqref{eqn-w}, \eqref{eqn_36} \eqref{eq35}, \eqref{eqn-es}, \eqref{eqn-a|e1_1}, and \eqref{eqn-a|e1_2} into \eqref{eqn-el}, the equation \eqref{eqn-ATL} can be obtained in Theorem \ref{theorem1}.

According to \eqref{tt}, the the expectation $\mathbb{E}\left[ {A_n^{\rm{t}} {T_n^{\rm{t,e}}}| E_n^{\rm{t,e}}} \right]$ can be expressed by
\begin{equation}\label{eatm}
\mathbb{E}\left[A_n^{\rm{t}}T_n^{\rm{t,e}}|E_n^{\rm{t,e}}\right]=\mathbb{E}\left[A_n^{\rm{t}}T_n^{\rm{t}}|E_n^{\rm{t,e}}\right]+\mathbb{E}\left[A_n^{\rm{t}}T_n^{\rm{e}}|E_n^{\rm{t,e}}\right],
\end{equation}
where the result of $\mathbb{E}\left[A_n^{\rm{t}}T_n^{\rm{t}}|E_n^{\rm{t,e}}\right]$ can be obtained similar to the derivation steps in \eqref{eq31}-\eqref{eqn-a|e1_2}, and thus is omitted. The result of $\mathbb{E}\left[A_n^{\rm{t}}T_n^{\rm{t}}|E_n^{\rm{t,e}}\right]$, defined by $\Xi_{i,j}^{2}$, is given in $\eqref{eqn-ATT}$ of Theorem \ref{theorem1}. Additionally, the derivation of $\mathbb{E}\left[A_n^{\rm{t}}T_n^{\rm{e}}|E_n^{\rm{t,e}}\right]$ is given as follows. Similarly, the system time $T_n^{\rm{e}}$ for computing the task at the edge server can be expressed by $T_n^{\rm{e}}=S_n^{\rm{e}}+W_n^{\rm{e}}$, where $S_n^{\rm{e}}$ and $W_n^{\rm{e}}$ are the service time and the waiting time of the $n$-th computation-intensive task at the edge server of BS, respectively. The waiting time $W_n^{\rm{e}}$ can be expressed by
\begin{equation}\label{wnr}
 		W_n^{\rm{e}} = \left\{
 		\begin{aligned}
 			&T_{n-1}^{\rm{e}}-A_n^{\rm{e}}, &T_{n-1}^{\rm{e}}>A_n^{\rm{e}}\\
 			&0,&T_{n-1}^{\rm{e}}<A_n^{\rm{e}},
 		\end{aligned}
 		\right.
\end{equation}
where $A_n^{\rm{e}}$ denotes the arrival interval between the $(n-1)$-th and $n$-th computation-intensive tasks at the edge server of BS. When the $n$-th task arrives at the transmitting queue of the transmitter, the $(n-1)$-th task is still waiting or served in the transmitting queue. Consequently, the arrival interval $A_n^{\rm{e}}$ between the $(n-1)$-th and $n$-th tasks at the edge server is equal to the server time $S_n^{\rm{t}}$ of the $n$-th task in the transmitting queue of the transmitter. Conversely, if the transmitting queue of the transmitter is idle when the $n$-th task arrives, the arrival interval $A_n^{\rm{e}}$ consists of both the service time $S_n^{\rm{t}}$ and the interval between leaving time of the $(n-1)$-th task and arriving time of the $n$-th task in the transmitting queue. Consequently, the arrival interval $A_n^{\rm{e}}$ can be expressed by
 	\begin{equation}\label{anr}
 		A_n^{\rm{e}} = \left\{
 		\begin{aligned}
 			&A_n^{\rm{t}}-T_{n-1}^{\rm{t}}+S_n^{\rm{t}}, &T_{n-1}^{\rm{t}}<A_n^{\rm{t}}\\
 			&S_n^{\rm{t}},&T_{n-1}^{\rm{t}}>A_n^{\rm{t}}.
 		\end{aligned}
 		\right.
 	\end{equation}
Consequently, $W_n^{\rm{e}}$ can be rewritten as
\begin{equation}\label{wnr_2}
	W_n^{\rm{e}}\! = \!\left\{\!
		\begin{aligned}
			&T_{n-1}^{\rm{e}}\!-\!A_n^{\rm{t}}\!+\!T_{n-1}^{\rm{t}}\!-\!S_n^{\rm{t}}, 		
			\!\!\!&T_{n-1}^{\rm{t}}\!<\!A_n^{\rm{t}} \,,\, T_{n-1}^{\rm{e}}\!>\!A_n^{\rm{e}}\\
			&T_{n-1}^{\rm{e}}\!-\!S_n^{\rm{t}},\!\!\!&T_{n-1}^{\rm{t}}\!>\!A_n^{\rm{t}}\, ,\, T_{n-1}^{\rm{e}}\!>\!A_n^{\rm{e}}\\
			&0,\!\!\!&\text{otherwise}.
		\end{aligned}
		\right.
\end{equation}
As a result, the expectation $\mathbb{E}\left[A_n^{\rm{t}}T_n^{\rm{e}}|E_n^{\rm{t,e}}\right]$ in \eqref{eatm} can be expressed by
\begin{equation}\label{eq56}
    \begin{aligned}
    	&\mathbb{E}\left[A_n^{\rm{t}}T_n^{\rm{e}}|E_n^{\rm{t,e}}\right]=\mathbb{E}\left[ {A_n^{\rm{t}}S_n^{\rm{e}}|E_n^{\rm{t,e}}} \right] + \mathbb{E}\left[ {A_n^{\rm{t}}W_n^{\rm{e}}|E_n^{\rm{t,e}}} \right]\\
    	&=\mathbb{E}\left[ {A_n^{\rm{t}}S_n^{\rm{e}}|E_n^{\rm{t,e}}} \right]+{\rm{P}}\left( {{J_n}} \right)\mathbb{E}\left[ {A_n^{\rm{t}}W_n^{\rm{e}}|T_n^{\rm{t,e}} > T_n^{\rm{l}},{J_n}} \right]\\
    	& + {\rm{P}}\left( {{I_n}} \right)\mathbb{E}\left[ {A_n^{\rm{t}}W_n^{\rm{e}}|T_n^{\rm{t,e}} > T_n^{\rm{l}},{I_n}} \right],
    \end{aligned}
\end{equation}
where $J_n$ and $I_n$ denote the events $\left(T_{n-1}^{\rm{t}}>A_n^{\rm{t}} \,,\, T_{n-1}^{\rm{e}}>A_n^{\rm{e}}\right)$ and $\left(T_{n-1}^{\rm{t}}<A_n^{\rm{t}} \,,\, T_{n-1}^{\rm{e}}>A_n^{\rm{e}}\right)$ in \eqref{wnr_2}, respectively, for denotational simplicity. The result of $\mathbb{E}\left[A_n^{\rm{t}}T_n^{\rm{e}}|E_n^{\rm{t,e}}\right]$ can be obtained similar to the derivation of $\mathbb{E}\left[A_n^{\rm{l}}T_n^{\rm{l}}|E_n^{\rm{l}}\right]$ in \eqref{eqn-el}. For denotational simplicity, the results of $\mathbb{E}\left[ {A_n^{\rm{t}}S_n^{\rm{e}}|E_n^{\rm{t,e}}} \right]$, $\mathbb{E}\left[ {A_n^{\rm{t}}W_n^{\rm{e}}|T_n^{\rm{t,e}} > T_n^{\rm{l}},{B_n}} \right]$, and $\mathbb{E}\left[ {A_n^{\rm{t}}W_n^{\rm{e}}|T_n^{\rm{t,e}} > T_n^{\rm{l}},{I_n}} \right]$ are defined as $\Xi_{i,j}^{3}$, $\Xi_{i,j}^{\rm{AWR1}}$, and $\Xi_{i,j}^{\rm{AWR2}}$ in $\eqref{eqn-ASR}$-$\eqref{eqn-AWR2}$ of Theorem \ref{theorem1}, respectively. 

As a result, by incorporating \eqref{eq_el}, \eqref{PEM}, and $\eqref{eqn-ATL}$-$\eqref{eqn-AWR2}$ into $\eqref{eqn-9}$, the result of MAoI under the partial offloading scheme in Theorem \ref{theorem1}.

\end{appendices}


\vfill

\end{document}